\documentclass{article}
\usepackage[utf8]{inputenc}
\usepackage{fullpage}
\usepackage{hyperref}
\usepackage{mathtools}
\DeclarePairedDelimiter\abs{\lvert}{\rvert}

\usepackage{amsfonts,amsmath}
\usepackage{soul}
\usepackage{amsthm}
\usepackage{mathrsfs}
\usepackage{authblk}
\usepackage{slashed}
\usepackage{dsfont}
\usepackage{cancel}

\DeclareMathOperator{\diag}{diag}
\DeclareMathOperator{\sgn}{sgn}
\DeclareMathOperator{\disc}{disc}

\newtheorem{thm}{Theorem}

\newtheorem{pro}{Proposition}
\newtheorem{lem}{Lemma}

\title{The Dirac equation across the horizons of the 5D Myers-Perry geometry : Separation of variables, radial asymptotic behaviour and Hamiltonian formalism}
\author{Qiu Shi Wang$^1$\footnote{Current affiliation (as of 04/2024) : Mathematical Institute, University of Oxford, Oxford, OX2 6GG, United Kingdom, wangqs@maths.ox.ac.uk}}
\affil{$^1$\small Department of Mathematics and Statistics\\ McGill University \\ Montréal, QC, H3A 2K6, Canada\\ qiu.s.wang@mail.mcgill.ca}
\date{3 April 2024}

\begin{document}

\maketitle
\begin{abstract}
We analytically extend the 5D Myers--Perry metric through the event and Cauchy horizons by defining Eddington--Finkelstein-type coordinates. Then, we use the orthonormal frame formalism to formulate and perform separation of variables on the massive Dirac equation, and analyse the asymptotic behaviour at the horizons and at infinity of the solutions to the radial ordinary differential equation (ODE) thus obtained. Using the essential self-adjointness result of Finster--Röken and Stone's formula, we obtain an integral spectral representation of the Dirac propagator for spinors with low masses and suitably bounded frequency spectra in terms of resolvents of the Dirac Hamiltonian, which can in turn be expressed in terms of Green's functions of the radial ODE.
\end{abstract}

\section{Introduction}

Detailed study of solutions of the massive Dirac equation in the Kerr geometry was made possible by Chandrasekhar's separation of variables procedure \cite{chandrasekhar}, in which the Dirac equation is separated into radial and angular systems of ordinary differential equations (ODEs). Finster, Kamran, Smoller and Yau use properties of the latter ODEs, notably the radial asymptotic behaviour of their solutions, to derive an integral spectral representation for the Dirac propagator \cite{Finster_2003} and study the long-term behaviour of Dirac particles \cite{Finster_2003, Finster_2002}. However, this analysis is performed on the Kerr(--Newman) metric in Boyer--Lindquist coordinates, which are singular at the event and Cauchy horizons. The results obtained in this coordinate system are thus valid only in the region outside the event horizon.

To study the Dirac equation in all regions of the Kerr geometry, Röken \cite{R_ken_2017} uses Eddington--Finkelstein-type coordinates, an analytic extension of the usual Boyer--Lindquist coordinates which is regular through the horizons, up until the ring singularity. However, an additional technical difficulty arises in the construction of an integral spectral representation for the Dirac propagator. Since the Dirac Hamiltonian is not elliptic at the horizons, standard results pertaining to elliptic operators cannot be applied to conclude its self-adjointness. To remedy this issue, Finster and Röken \cite{Finster_2016} construct a self-adjoint extension of the Dirac Hamiltonian on a class of Lorentzian spin manifolds, in particular the Kerr geometry in Eddington--Finkelstein-type coordinates with an additional inner radial boundary inside the Cauchy horizon. This allows them in \cite{Finster_2018} to use Stone's formula \cite{reed_simon_1978} to construct an integral spectral representation for the Dirac propagator in terms of the Green's matrix of the radial ODE obtained by \cite{R_ken_2017} using the Newman--Penrose formalism. Their construction involves spectral projectors onto the eigenspaces of the angular Dirac operator. However, Stone's formula is expressed in terms of slightly complex frequencies, in which case it is unclear whether such spectral projectors exist, as the angular Dirac operator is no longer self-adjoint.

Derived by Myers and Perry \cite{MYERS1986304}, the Myers--Perry metrics generalize the Kerr metric to higher dimensions. In particular, the 5-dimensional Myers--Perry geometry describes a black hole rotating in two independent directions. In contrast to Röken \cite{R_ken_2017} using a null frame in the Kerr geometry, Wu \cite{Wu_2008} uses the orthonormal frame formalism to formulate the Dirac equation in the 5D Myers--Perry geometry in Boyer--Lindquist coordinates, then separates it into radial and angular ODEs. It is possible \cite{BLpaper} to use the latter equations to derive an integral spectral representation for the Dirac propagator in the exterior region of this geometry using the methods of \cite{Finster_2003}.

A problem of interest would therefore be to study the separability and radial asymptotics of the Dirac equation through the horizons of the 5D Myers--Perry geometry using a suitable analogue of Eddington--Finkelstein-type coordinates in 5 dimensions and the orthonormal frame formalism. As the self-adjointness result of \cite{Finster_2016} is not specific to 4 dimensions, it is natural to ask whether it may be applied in these coordinates to construct an integral spectral representation for the Dirac propagator, analogously to \cite{Finster_2018}. To address these questions is the objective of the present work. In the case of sufficiently small (spinor) masses and of frequencies satisfying a specific mass-dependent bound, it is possible to resolve the aforementioned difficulty in constructing spectral projectors for the angular Dirac operator by using the method of slightly non-self-adjoint perturbations of \cite{Finster_2006}. A limitation of this approach is that we may only obtain an integral spectral representation for the Dirac propagator applied to initial data with small masses and with frequency spectra contained in the above-mentioned mass-dependent bounded interval. In other words, to obtain an integral spectral representation for the solution to the Cauchy problem with arbitrary initial data, one would need to show the existence of angular spectral projectors for all masses and frequencies, which are not available via this method.

The paper is organised as follows. In Section \ref{5DMP}, we introduce the 5D Myers--Perry geometry in Boyer--Lindquist coordinates, and in Section \ref{EF} we derive Eddington--Finkelstein-type coordinates which are regular across the horizons. In Section \ref{Cartan}, we use a regular orthonormal frame to explicitly formulate the Dirac equation. The construction of the regular frame from the Kinnersley-type Newman--Penrose vectors of \cite{Wu_2008}, as well as details of the computation of the Dirac operator, can be found in Appendix \ref{frameappendix}. We then use an ansatz similar to Chandrasekhar's ansatz \cite{chandrasekhar} to separate the Dirac equation into radial and angular ODEs in Section \ref{sepvar}. The asymptotic behaviour of solutions to the radial ODE is described by Propositions \ref{eventasympt}, \ref{Cauchyasympt} and \ref{infinityasympt} in Section \ref{asymptotics}. In Section \ref{self-adj}, we establish the essential self-adjointness of the Dirac Hamiltonian with a suitable domain of definition using the results of \cite{Finster_2016}, and in Section \ref{angularprojectors} we construct spectral projectors for the angular Dirac operator with slightly complex frequencies for small masses and suitably bounded frequencies using the method of slightly non-self-adjoint perturbations \cite{Finster_2006}. Finally, we construct in Section \ref{intrep} an integral spectral representation for the Dirac propagator for initial data subject to the above restrictions on the mass and frequency spectrum in terms of Green's functions of the radial ODE using a variant of Stone's formula.

\section{The 5D Myers--Perry black hole in Boyer--Lindquist coordinates}\label{5DMP}

In this section, we introduce the 5D Myers--Perry black hole in Boyer--Lindquist coordinates, well-defined outside of the event horizon.

In Boyer--Lindquist coordinates $(t,\rho,\vartheta,\varphi,\psi)$, a 5D Myers--Perry black hole is represented by the manifold
\begin{equation*}
\mathcal{M}=\mathbb{R}_t\times (0,\infty)_\rho \times (0,\frac{\pi}{2})_\vartheta \times [0,2\pi)_\varphi \times [0,2\pi)_\psi
\end{equation*}
equipped with the Lorentzian metric
\begin{equation}\label{metric}
g=-dt^2+\frac{\Sigma \rho^2}{\Delta}d\rho^2+\Sigma d\vartheta^2 + (r^2+a^2)\sin^2\vartheta d\varphi^2 + (r^2+b^2)\cos^2\vartheta d\psi^2 + \frac{\mu}{\Sigma}(dt-a\sin^2\vartheta d\varphi - b\cos^2\vartheta d\psi)^2,
\end{equation}
where $\mu/2$ is the mass of the black hole, $a$ and $b$ its two independent angular momenta,
\begin{equation*}
\Delta=(\rho^2+a^2)(\rho^2+b^2)-\mu \rho^2,\qquad \Sigma=\rho^2+a^2\cos^2\vartheta+b^2\sin^2\vartheta.
\end{equation*}
The Lorentzian manifold $(\mathcal{M},g)$ solves the vacuum Einstein equations in 5 dimensions, i.e. it is Ricci flat. Furthermore, it has three commuting Killing vector fields $\partial_t$, $\partial_\varphi$ and $\partial_\psi$. We restrict our attention to the non-extreme case $\mu>a^2+b^2+2\abs{ab}$, for which $\Delta(\rho)$ has two distinct positive roots
\begin{equation*}
\rho^2_\pm = \frac{1}{2}\big(\mu -a^2-b^2\pm \sqrt{(\mu-a^2-b^2)^2-4a^2b^2}\big).
\end{equation*}
The radii $\rho_-,\rho_+$ are those of the Cauchy and event horizons respectively. 

\section{Eddington--Finkelstein-type coordinates}\label{EF}

In this section, we derive Eddington--Finkelstein-type coordinates for the 5D Myers--Perry black hole.  We also derive the transformation laws between Boyer--Lindquist and Eddington--Finkelstein-type coordinates for the coordinate 1-forms and vector fields, for use in subsequent sections.

The 5D Myers--Perry metric in Boyer--Lindquist coordinates (\ref{metric}) is singular at the event and Cauchy horizons, as $g_{\rho\rho}=\Sigma \rho^2/\Delta\rightarrow \infty$ as $\Delta\rightarrow 0$. Since the black hole has coordinate singularities rather than curvature singularities at the horizons, it is possible to construct a system of coordinates that analytically extends the 5D Myers--Perry metric through them. This results in the metric (\ref{EFmetric}), which is fully regular up until the essential singularity at $r=0$.

The 5D Myers--Perry geometry has a pair of real principal null geodesic vector fields \cite{Daud__2012} which, in Boyer--Lindquist coordinates, are of the form
\begin{equation*}
V^\pm=\frac{(\rho^2+a^2)(\rho^2+b^2)}{\Delta}\left(\partial_t+\frac{a}{\rho^2+a^2}\partial_\varphi+\frac{b}{\rho^2+b^2}\partial_\psi\right)\pm \partial_\rho.
\end{equation*}
The tangent vectors to the principal null geodesics satisfy, in terms of an affine parameter $\lambda$,
\begin{equation}\label{tangentvector}
\frac{dt}{d\lambda} = \frac{(\rho^2+a^2)(\rho^2+b^2)}{\Delta}C,\quad \frac{d\rho}{d\lambda}=\pm C,\quad \frac{d\theta}{d\lambda}=0, \quad \frac{d\varphi}{d\lambda} = \frac{a(\rho^2+b^2)}{\Delta}C,\quad \frac{d\psi}{d\lambda} = \frac{b(\rho^2+a^2)}{\Delta}C.
\end{equation}
We will use the ``Regge--Wheeler'' radial coordinate $x$ defined by
\begin{equation}\label{RWcoord}
\frac{dx}{d\rho}=\frac{(\rho^2+a^2)(\rho^2+b^2)}{\Delta}.
\end{equation}
Following the form of (\ref{tangentvector}), we then define the Eddington--Finkelstein-type coordinates $(\tau,r,\theta,\phi,\xi)$ on the 5D Myers--Perry black hole by the coordinate transformation
\begin{equation}\label{coordchange}
\begin{split}
\tau&=t+x-\rho\\
r&=\rho\\
\theta&=\vartheta\\
\phi&=\varphi+\int \frac{a(\rho^2+b^2)}{\Delta}\; d\rho\\
\xi&=\psi + \int \frac{b(\rho^2+a^2)}{\Delta}\; d\rho.
\end{split}
\end{equation}
In terms of the coordinate 1-forms, the above change of variables reads
\begin{equation*}
\begin{split}
d\tau&=dt+\frac{\mu\rho^2}{\Delta}d\rho\\
dr&=d\rho\\
d\theta&=d\vartheta\\
d\phi&=d\varphi+\frac{a(\rho^2+b^2)}{\Delta}d\rho\\
d\xi&=d\psi+\frac{b(\rho^2+a^2)}{\Delta}d\rho.
\end{split}
\end{equation*}
The metric (\ref{metric}) can thus be written as
\begin{multline}\label{EFmetric}
g=\left( -1+\frac{\mu}{\Sigma}\right) d\tau^2 + \frac{2\mu}{\Sigma}d\tau(dr-a\sin^2\theta d\phi -b\cos^2\theta d\xi) + \left(1+\frac{\mu}{\Sigma}\right) (dr-a\sin^2\theta d\phi - b\cos^2\theta d\xi)^2\\
+\Sigma d\theta^2+(r^2+a^2)\sin^2\theta d\phi^2 + (r^2+b^2)\cos^2\theta d\xi^2 - (a\sin^2\theta d\phi + b\cos^2\theta d\xi)^2.
\end{multline}
The above metric, regular across the horizons, is defined on the Lorentzian manifold $(\mathcal{N},g)$, where
\begin{equation*}
\mathcal{N} =\mathbb{R}_\tau\times (0,\infty)_r \times (0,\frac{\pi}{2})_\theta \times [0,2\pi)_\phi \times [0,2\pi)_\xi.
\end{equation*}
As the induced metric on constant-$\tau$ hypersurfaces is positive definite, $\tau$ is a proper time function. Furthermore, we note that the coordinate vector fields transform as, keeping the same notation for $\Delta(\rho)=\Delta(r)$ and $\Sigma(\rho,\vartheta)=\Sigma(r,\theta)$,
\begin{equation}\label{changevectorfields}
\begin{split}
\partial_t&=\partial_\tau\\
\partial_\rho&=\partial_r+\frac{\mu r^2}{\Delta}\partial_\tau + \frac{a(r^2+b^2)}{\Delta}\partial_\phi + \frac{b(r^2+a^2)}{\Delta}\partial_\xi\\
\partial_\vartheta&=\partial_\theta\\
\partial_\varphi&=\partial_\phi\\
\partial_\psi&=\partial_\xi.
\end{split}
\end{equation}

\section{Orthonormal frame formalism for the Dirac equation}\label{Cartan}

In the absence of a suitable Newman--Penrose formalism in 5 dimensions, we resort to using the equivalent but more computationally tedious orthonormal frame formalism for the Dirac equation. In this section, we first introduce the Dirac equation and the choice of gamma matrices used. Then, we use the regular orthonormal frame (\ref{framevectors}),(\ref{frame1forms}) constructed in Appendix \ref{frameappendix} to obtain an explicit expression for the Dirac operator, noting that its angular part is identical in Eddington--Finkelstein-type and Boyer--Lindquist coordinates \cite{Wu_2008,Daud__2012}.

We choose the gamma matrices $\gamma^A$, $A=0,1,2,3,5$, as
\begin{equation*}
\gamma^0=i\begin{pmatrix}
0 & I \\
I & 0
\end{pmatrix},\enspace
\gamma^1=i\begin{pmatrix}
0 & \sigma^3 \\
-\sigma^3 & 0
\end{pmatrix},\enspace
\gamma^2=i\begin{pmatrix}
0 & \sigma^1\\
-\sigma^1 & 0
\end{pmatrix},\enspace
\gamma^3=i\begin{pmatrix}
0 & \sigma^2\\
-\sigma^2 & 0
\end{pmatrix},\enspace
\gamma^5=\begin{pmatrix}
I & 0\\
0 & -I
\end{pmatrix},
\end{equation*}
where the $\sigma^j$ are the Pauli matrices
\begin{equation}\label{paulimatrices}
    \sigma_1 = \begin{pmatrix}
        0 & 1 \\
        1 & 0
    \end{pmatrix},\quad \sigma_2 = \begin{pmatrix}
        0 & -i\\
        i & 0
    \end{pmatrix}, \quad \sigma_3 = \begin{pmatrix}
        1 & 0\\
        0 & -1
    \end{pmatrix}.
\end{equation}
They satisfy the Clifford algebra anticommutation relations
\begin{equation}\label{gammaanticom}
\{\gamma^A,\gamma^B\}=2\eta^{AB},
\end{equation}
where $\eta^{AB}=\diag\{-1,1,1,1,1\}$ is the five-dimensional Minkowski metric in our chosen signature. We define the matrices $\Gamma^A$ by $\Gamma^0=i\gamma^0$ and $\Gamma^j=-\gamma^0\gamma^j$ for $j\neq 0$; they satisfy the anticommutation relations
\begin{equation*}
\{\Gamma^A,\Gamma^B\}=2\delta^{AB},
\end{equation*}
where $\delta^{AB}=\diag\{1,1,1,1,1\}$ is the Kronecker delta.

The massive Dirac equation takes the form
\begin{equation}\label{dirac}
(\gamma^A(\partial_A+\Gamma_A)-m)\psi =0,
\end{equation}
where $\Gamma_A$ are the components of the spinor connection  $\Gamma=\Gamma_Ae^A=\frac{1}{4}\gamma^A\gamma^B\omega_{AB}$ in an orthonormal pentad frame $e^A={e^A}_\mu dx^\mu$, and $\omega_{AB}$ is the connection 1-form in the same frame. Cartan's first structure equation then relates the components of ${\omega^A}_B=\eta^{AC}\omega_{CB}$ and the orthonormal frame $e^A$ via
\begin{equation}\label{cartanfirst}
de^A = -{\omega^A}_B \wedge e^B.
\end{equation}

In the 5D Myers--Perry geometry extended through the horizons, one can find the regular orthonormal frame (\ref{framevectors}),(\ref{frame1forms}). In this frame, we can compute ${\omega^A}_B$, from which we obtain the coefficients $\Gamma_A$ of the spinor connection $\Gamma_Ae^A=\frac{1}{2}\sum_{A<B}\gamma^A\gamma^B\omega_{AB}$. Details of the construction of the frame and explicit formulae for the connection coefficients can be found in Appendix \ref{frameappendix}. Using the anticommutation relations (\ref{gammaanticom}), the relation $\gamma^5 = -i\gamma^0\gamma^1\gamma^2\gamma^3$ and the form of the orthonormal frame (\ref{framevectors}), the Dirac operator can therefore be written, after various simplifications, as
\begin{multline*}
\gamma^A(\partial_A+\Gamma_A) = \gamma^0 \frac{1}{r_+^3\sqrt{\Sigma}}\Biggl(\left(\frac{\Delta}{2}+2Mr^2+\frac{r_+^6}{2r^2}\right)\partial_\tau + \left(\frac{\Delta}{2}-\frac{r_+^6}{2r^2}\right)\partial_r + a(r^2+b^2)\partial_\phi + b(r^2+a^2)\partial_\xi \\
+ \frac{\partial_r\Delta}{4}+\frac{\Delta}{4r} + \frac{r_+^6}{4r^3} + \left(\frac{\Delta}{2}-\frac{r_+^6}{2r^2}\right)\frac{r-ip\gamma^5}{2\Sigma}\Biggr)\\
+\gamma^1 \frac{1}{r_+^3\sqrt{\Sigma}}\Biggl(\left(\frac{\Delta}{2}+2Mr^2-\frac{r_+^6}{2r^2}\right)\partial_\tau + \left(\frac{\Delta}{2}+\frac{r_+^6}{2r^2}\right)\partial_r + a(r^2+b^2)\partial_\phi + b(r^2+a^2)\partial_\xi \\
+ \frac{\partial_r\Delta}{4}+\frac{\Delta}{4r} - \frac{r_+^6}{4r^3} + \left(\frac{\Delta}{2}+\frac{r_+^6}{2r^2}\right)\frac{r-ip\gamma^5}{2\Sigma}\Biggr)\\
+\gamma^2\frac{1}{\sqrt{\Sigma}}\left(\partial_\theta + \frac{\cot\theta}{2} - \frac{\tan\theta}{2} - \frac{(a^2-b^2)\sin\theta\cos\theta}{2p\Sigma}i\gamma^5 (r-ip\gamma^5)\right) \\
+ \gamma^3 \frac{\sin\theta\cos\theta}{p\sqrt{\Sigma}}\left((a^2-b^2)\partial_\tau + \frac{a}{\sin^2\theta}\partial_\phi - \frac{b}{\cos^2\theta}\partial_\xi\right) \\
+ \gamma^5\frac{1}{rp}(ab\partial_\tau + b\partial_\phi + a\partial_\xi) + \gamma^0\gamma^1\frac{iab}{2r^2p^2}(r+ip\gamma^5),
\end{multline*}
where $p=\sqrt{a^2\cos^2\theta + b^2\sin^2\theta}$.

\section{Separation of variables}\label{sepvar}
In this section, we follow \cite{Daud__2012} and perform an invertible, time-independent transformation of spinors $\psi \mapsto \psi'=\mathscr{P}\psi$. We then put the Dirac equation into Hamiltonian form and separate it into radial and angular ODEs. The angular ODE thus obtained is the same as in Boyer--Lindquist coordinates \cite{BLpaper}.

Let $\mathscr{P}$ be a square root of $r+ip\gamma^5$, for instance 
\begin{equation}\label{Ptransform}
\mathscr{P}=\sqrt{\frac{r+\sqrt{\Sigma}}{2}} + i\sqrt{\frac{\sqrt{\Sigma}-r}{2}}\gamma^5.
\end{equation}
In terms of $\psi'=\mathscr{P}\psi$, the Dirac equation (\ref{dirac}) is then 
\begin{multline}\label{transformeddirac}
\Biggl[ \left( \frac{1}{r_+^3} \left( \gamma^0\left(\frac{\Delta}{2}+2Mr^2+\frac{r_+^6}{2r^2}\right) + \gamma^1\left(\frac{\Delta}{2}+2Mr^2-\frac{r_+^6}{2r^2}\right)\right)+\gamma^3\frac{\sin\theta\cos\theta}{p}(a^2-b^2) + \left(\frac{\gamma^5}{p} - \frac{i}{r}\right)ab\right)\partial_\tau\\
+\frac{1}{r_+^3}\left(\gamma^0\left(\frac{\Delta}{2}-\frac{r_+^6}{2r^2}\right)+\gamma^1\left(\frac{\Delta}{2}+\frac{r_+^6}{2r^2}\right)\right)\partial_r + \gamma^2\left(\partial_\theta + \frac{\cot\theta}{2}-\frac{\tan\theta}{2}\right) + \gamma^3\frac{1}{p}(a\cot\theta\partial_\phi - b\tan\theta\partial_\xi)\\
+ \gamma^5\frac{1}{p}(b\partial_\phi + a\partial_\xi) + \left((\gamma^0+\gamma^1)\frac{a(r^2+b^2)}{r_+^3}-\frac{ib}{r}\right)\partial_\phi + \left((\gamma^0+\gamma^1)\frac{b(r^2+a^2)}{r_+^3}-\frac{ia}{r}\right)\partial_\xi + \gamma^0\gamma^1\frac{iab}{r^2}\\
+\gamma^0\frac{1}{r_+^3}\left(\frac{\partial_r\Delta}{4}+\frac{\Delta}{4r} + \frac{r_+^6}{4r^3}\right) + \gamma^1\frac{1}{r_+^3}\left(\frac{\partial_r\Delta}{4}+\frac{\Delta}{4r} - \frac{r_+^6}{4r^3}\right) - mr + imp\gamma^5\Biggr] \psi' = 0.
\end{multline}
Denote the prefactor of $\partial_\tau$ by
\begin{equation*}
P\equiv\gamma^0\left(\frac{\Delta}{2r_+^3} + \frac{2Mr^2}{r_+^3} + \frac{r_+^3}{2r^2}\right) + \gamma^1\left(\frac{\Delta}{2r_+^3} + \frac{2Mr^2}{r_+^3} - \frac{r_+^3}{2r^2}\right) + \gamma^3 \frac{\sin\theta\cos\theta}{p}(a^2-b^2) + \left(\frac{\gamma^5}{p} - \frac{i}{r}\right)ab.
\end{equation*}
It is invertible, with inverse
\begin{multline*}
P^{-1}=\frac{-1}{\Sigma+2M}\biggl(\gamma^0\left(\frac{\Delta}{2r_+^3} + \frac{2Mr^2}{r_+^3} + \frac{r_+^3}{2r^2}\right) + \gamma^1\left(\frac{\Delta}{2r_+^3} + \frac{2Mr^2}{r_+^3} - \frac{r_+^3}{2r^2}\right) \\
+ \gamma^3 \frac{\sin\theta\cos\theta}{p}(a^2-b^2) + \left(\frac{\gamma^5}{p} + \frac{i}{r}\right)ab\biggr).
\end{multline*}
Left-multiplying (\ref{transformeddirac}) by $-iP^{-1}$ and defining $D_j = -i\partial_j$, one obtains
\begin{multline*}
i\partial_\tau \psi' = P^{-1} \Biggl[\frac{1}{r_+^3}\left(\gamma^0\left(\frac{\Delta}{2}-\frac{r_+^6}{2r^2}\right)+\gamma^1\left(\frac{\Delta}{2}+\frac{r_+^6}{2r^2}\right)\right)D_r + \gamma^0 \mathbb{D}_{S^3} + \\
\left((\gamma^0+\gamma^1)\frac{a(r^2+b^2)}{r_+^3}-\frac{ib}{r}\right)D_\phi + \left((\gamma^0+\gamma^1)\frac{b(r^2+a^2)}{r_+^3}-\frac{ia}{r}\right)D_\xi \\+
 \gamma^0\gamma^1 \frac{ab}{r^2} - \gamma^0 \frac{i}{r_+^3}\left(\frac{\partial_r\Delta}{4}+\frac{\Delta}{4r} + \frac{r_+^6}{4r^3}\right) - \gamma^1\frac{i}{r_+^3}\left(\frac{\partial_r\Delta}{4}+\frac{\Delta}{4r} - \frac{r_+^6}{4r^3}\right)\Biggr]\psi',
\end{multline*}
where 
\begin{equation}\label{DS3}
\mathbb{D}_{S^3}=i\gamma^0 \gamma^2 \Big( \partial_\theta + \frac{\cot\theta}{2}-\frac{\tan\theta}{2}\Big) + i\gamma^0\gamma^3 \frac{1}{\sin\theta}\partial_\phi+i\gamma^0\gamma^5\frac{1}{\cos\theta}\partial_\xi
\end{equation}
is the Dirac operator on $S^3$ \cite{Daud__2012}. For convenience, we make the substitution $P^{-1}=-N\gamma^0/r$, and thus obtain the Hamiltonian form of the Dirac equation
\begin{equation}\label{hamform}
i\partial_\tau \psi' = N\mathcal{D}_0\psi',
\end{equation}
where
\begin{multline*}
\mathcal{D}_0=\left(\frac{1}{rr_+^3}\left(\frac{\Delta}{2}-\frac{r_+^6}{2r^2}\right) + \Gamma^1\frac{1}{rr_+^3}\left(\frac{\Delta}{2}+\frac{r_+^6}{2r^2}\right)\right)D_r + \frac{1}{r}\mathbb{D}_{S^3} \\
+\left((\mathds{1}_4+\Gamma^1)\frac{a(r^2+b^2)}{rr_+^3}+\gamma^0\frac{ib}{r^2}\right)D_\phi
+\left((\mathds{1}_4+\Gamma^1)\frac{b(r^2+a^2)}{rr_+^3}+\gamma^0\frac{ia}{r^2}\right)D_\xi + \gamma^1 \frac{ab}{r^3} \\
- \frac{i}{r_+^3}\left(\frac{\partial_r\Delta}{4r}+\frac{\Delta}{4r^2} + \frac{r_+^6}{4r^4}\right) -\Gamma^1 \frac{i}{r_+^3}\left(\frac{\partial_r\Delta}{4r}+\frac{\Delta}{4r^2} - \frac{r_+^6}{4r^4}\right) - \gamma^0 im + \Gamma^5 \frac{mp}{r}.
\end{multline*}

Using the ansatz \cite{BLpaper}
\begin{equation}\label{ansatz}
\psi'(\tau,r,\theta,\phi,\xi)=e^{-i\omega \tau}e^{-i((k_a+\frac{1}{2})\phi+(k_b+\frac{1}{2})\xi)}\begin{pmatrix}
X_+(r) Y_+(\theta)\\
X_-(r) Y_-(\theta)\\
X_-(r) Y_+(\theta)\\
X_+(r) Y_-(\theta)
\end{pmatrix},
\end{equation}
where $\omega\in\mathbb{R}$ is the frequency of the Dirac particle and $k_a,k_b\in\mathbb{Z}$ its ``azimuthal quantum numbers'' along the two independent axes of rotation $\partial_\phi$, $\partial_\xi$ respectively, one can write the Dirac equation as
\begin{equation}\label{diracsep}
(\mathcal{D}_0-\omega N^{-1})\psi'=0,
\end{equation}
where
\begin{multline}\label{Ninverse}
N^{-1}=-\frac{\gamma^0P}{r}\\
=\left(\frac{\Delta}{2rr_+^3}+\frac{2Mr}{r_+^3}+\frac{r_+^3}{2r^3}\right) + \Gamma^1\left(\frac{\Delta}{2rr_+^3}+\frac{2Mr}{r_+^3}-\frac{r_+^3}{2r^3}\right) + \Gamma^3 \frac{\sin\theta\cos\theta(a^2-b^2)}{rp} + \Gamma^5\frac{ab}{rp} + \gamma^0 \frac{iab}{r^2}.
\end{multline}
The Dirac equation (\ref{diracsep}) can be separated into purely angular and radial parts
\begin{equation*}
\frac{1}{r}(\mathcal{R}+\mathcal{A})\psi'=0,
\end{equation*}
where
\begin{equation}\label{angularoperator}
\mathcal{A}=\mathbb{D}_{S^3}+mp\Gamma^5-\omega\Big(\frac{(a^2-b^2)\sin\theta\cos\theta}{p}\Gamma^3+\frac{ab}{p}\Gamma^5\Big)
\end{equation}
as in Boyer--Lindquist coordinates \cite{Daud__2012}, and 
\begin{multline*}
\mathcal{R}=\frac{1}{r_+^3}\left(\left(\frac{\Delta}{2}-\frac{r_+^6}{2r^2}\right) + \Gamma^1\left( \frac{\Delta}{2} + \frac{r_+^6}{2r^2}\right)\right)D_r - mr\Gamma^0 + \left( (I+\Gamma^1) \frac{a(r^2+b^2}{r_+^3} + \Gamma^0 \frac{b}{r}\right)D_\phi \\
+ \left( (I+\Gamma^1)\frac{b(r^2+a^2)}{r_+^3} + \Gamma^0 \frac{a}{r}\right)D_\xi + \frac{ab}{r^2}\gamma^1 - \frac{i}{r_+^3}\left(\frac{\partial_r\Delta}{4}+\frac{\Delta}{4r} + \frac{r_+^6}{4r^3}\right) - \Gamma^1 \frac{i}{r_+^3}\left(\frac{\partial_r\Delta}{4}+\frac{\Delta}{4r} - \frac{r_+^6}{4r^3}\right)\\
-\omega\left(\left(\frac{\Delta}{2r_+^3} + \frac{2Mr^2}{r_+^3} + \frac{r_+^3}{2r^2}\right) + \Gamma^1\left(\frac{\Delta}{2r_+^3} + \frac{2Mr^2}{r_+^3} - \frac{r_+^3}{2r^2}\right) + \frac{ab}{r}\Gamma^0\right).
\end{multline*}
As in \cite{BLpaper}, the ansatz (\ref{ansatz}) allows the replacements $D_\phi \rightarrow -(k_a+1/2)$ and $D_\xi \rightarrow -(k_b+1/2)$, under which the Dirac equation separates into the ODEs 
\begin{equation*}
-\mathcal{R}\psi' = \mathcal{A}\psi' = \gamma^1 \lambda\psi'
\end{equation*}
for a separation constant $\lambda$. Explicitly, the radial ODE is
\begin{multline}\label{radialraw}
\left( \frac{r_+^3}{r^2}\partial_r - \frac{r_+^3}{2r^3} + \frac{i\omega r_+^3}{r^2}\right) X_+ = \left( -\lambda +\frac{i}{r}(mr^2+a(k_a+\frac{1}{2})+b(k_b+\frac{1}{2}) + \omega ab) - \frac{ab}{r^2}\right)X_-,\\
\left(\frac{\Delta}{r_+^3}\partial_r - \frac{2i}{r_+^3}((k_a+\frac{1}{2})a(r^2+b^2) + (k_b+\frac{1}{2})b(r^2+a^2)) + \frac{\partial_r \Delta}{2r_+^3} + \frac{\Delta}{2rr_+^3}- \frac{i\omega}{r_+^3}(\Delta+4Mr^2)\right)X_-\\
=\left( -\lambda -\frac{i}{r}(mr^2+a(k_a+\frac{1}{2})+b(k_b+\frac{1}{2}) + \omega ab) - \frac{ab}{r^2}\right)X_+.
\end{multline}

We remark that (\ref{radialraw}) resembles the analogous result for the Kerr geometry in Eddington--Finkelstein-type coordinates \cite{Finster_2018}, namely
\begin{equation*}
\begin{split}
\frac{1}{r_+}\left( \Delta_K \partial_r + \frac{\partial_r \Delta_K}{2} - i\omega(\Delta_K + 4Mr) - 2iak\right) X_+ &=(\lambda + imr)X_-\\
r_+(\partial_r + i\omega) X_- &= (\lambda - imr)X_+,
\end{split}
\end{equation*}
where $\Delta_K=(r^2+a^2)-2Mr$ is the horizon-defining function for the Kerr geometry, vanishing at the event and Cauchy horizons.

Making the transformation $\tilde{X}_-=r\sqrt{\abs{\Delta}}X_-$ and $\tilde{X}_+ = r_+^3 X_+$, (\ref{radialraw}) takes the more symmetric form
\begin{multline}\label{radial}
\left( \Delta \partial_r - \frac{\Delta}{2r} + \Delta i\omega\right) \tilde{X}_+ = \sgn(\Delta)\sqrt{\abs{\Delta}} \left( -\lambda r +i(mr^2+a(k_a+\frac{1}{2})+b(k_b+\frac{1}{2}) + \omega ab) - \frac{ab}{r}\right)\tilde{X}_-,\\
\left(\Delta\partial_r - \frac{\Delta}{2r} -2i((k_a+\frac{1}{2})a(r^2+b^2) + (k_b+\frac{1}{2})b(r^2+a^2)) - i\omega(\Delta + 4Mr^2)\right)\tilde{X}_-\\
=\sqrt{\abs{\Delta}}\left( -\lambda r -i(mr^2+a(k_a+\frac{1}{2})+b(k_b+\frac{1}{2}) + \omega ab) - \frac{ab}{r}\right)\tilde{X}_+.
\end{multline}
The angular ODE is, exactly as in Boyer--Lindquist coordinates \cite{BLpaper},
\begin{equation}\label{angular}
AY\equiv
\begin{pmatrix}
C_b(\theta) & L_\theta+C_a(\theta)\\
-L_\theta + C_a(\theta) & -C_b(\theta)
\end{pmatrix}
\begin{pmatrix}Y_+\\
Y_-
\end{pmatrix} = \lambda 
\begin{pmatrix}Y_+\\
Y_-
\end{pmatrix},
\end{equation}
where
\begin{align}\label{angularabbrevs}
\begin{split}
L_\theta&=\partial_\theta + \frac{\cot\theta}{2}-\frac{\tan\theta}{2}\\
C_a(\theta)&=-\frac{(k_a+\frac{1}{2})}{\sin\theta}-\frac{\omega(a^2-b^2)\sin\theta\cos\theta}{p}\\
C_b(\theta)&=-\frac{(k_b+\frac{1}{2})}{\cos\theta}+mp-\frac{\omega ab}{p}.
\end{split}
\end{align}
As in the Kerr--Newman geometry \cite{Finster_2003}, the angular eigenvalues satisfy certain nondegeneracy and regularity properties. The following was shown in \cite[Appendix A]{BLpaper}.

\begin{pro}\label{angularnondegenerate}
For any $\omega\in\mathbb{R}$ and $k_a,k_b\in \mathbb{Z}$, the angular operator $A$ (\ref{angular}) has discrete, real, nondegenerate spectrum $\sigma(A^{k_ak_b\omega}) = \{\lambda_l\}$ and eigenvectors $\left(Y_{+,l}^{k_ak_b},Y_{-,l}^{k_ak_b}\right)$, where $l\in\mathbb{Z}$, both smoothly dependent on the frequency $\omega$. As functions of $\omega$, the eigenvalues can thus be ordered $\lambda_l(\omega)<\lambda_{l+1}(\omega)$.
\end{pro}

\section{Asymptotic behaviour of solutions of the radial ODE}\label{asymptotics}
Following \cite{Finster_2003}, we study in this section the asymptotic behaviour of solutions of the radial ODE (\ref{radial}) near the event and Cauchy horizons and at radial infinity.

In terms of the radial coordinate $x$ defined in (\ref{RWcoord}) and $\tilde{X}=(\tilde{X}_+,\tilde{X}_-)$, the radial ODE may be rewritten
\begin{equation*}
\partial_x \tilde{X}=\frac{1}{(r^2+a^2)(r^2+b^2)}\begin{pmatrix} \frac{\Delta}{2r}-\Delta i\omega & -r^{-1}\sgn(\Delta)\sqrt{\abs{\Delta}}\; S(r) \\
-r^{-1}\sqrt{\abs{\Delta}}\; \overline{S}(r) & \frac{\Delta}{2r}+2i U(r)+i\omega(\Delta + 4Mr^2)\end{pmatrix}\tilde{X},
\end{equation*}
where
\begin{align}
U(r)&=\left(k_a+\frac{1}{2}\right)a(r^2+b^2)+\left(k_b+\frac{1}{2}\right)b(r^2+a^2)\label{Ur}\\
S(r)&=\lambda + i\left(mr+\frac{a}{r}\left(k_a+\frac{1}{2}\right)+\frac{b}{r}\left(k_b+\frac{1}{2}\right) + \frac{\omega ab}{r}\right) + \frac{ab}{r^2}.\label{Sr}
\end{align}
The asymptotic behaviour of $\tilde{X}$ at the horizons is similar to the case of the Kerr geometry in Eddington--Finkelstein-type coordinates \cite{R_ken_2017}, but different from the case of the 5D Myers--Perry geometry in Boyer--Lindquist coordinates \cite{BLpaper}. Namely, we have

\begin{pro}\label{eventasympt}
Every nontrivial solution $\tilde{X}$ of (\ref{radial}) is asymptotically as $r\searrow r_+$ of the form

\begin{equation*}
\tilde{X}(x)=\begin{pmatrix}
g_{+,1}\\
g_{+,2}e^{2i\Omega_+ x}
\end{pmatrix} + R_+(x)
\end{equation*}
with
\begin{align}
\abs{g_{+,1}}^2 + \abs{g_{+,2}}^2 &\neq 0\nonumber \\ 
\Omega_+ &= \omega + \frac{a(k_a+\frac{1}{2})}{r_+^2+a^2} + \frac{b(k_b+\frac{1}{2})}{r_+^2+b^2}\label{omegaplus}\\
\abs{R_+}&\leq c_+ e^{d_+ x} \nonumber
\end{align}
for suitable constants $c_+,d_+>0$.
\end{pro}

\begin{proof}
We make the ansatz
\begin{equation*}
\tilde{X}(x)=\begin{pmatrix}
g_1(x)\\
g_2(x)e^{2i\Omega_+ x}
\end{pmatrix}.
\end{equation*}
Then, by direct computation, $g(x)=(g_1(x),g_2(x))$ satisfies the differential equation
\begin{equation*}
\partial_x \begin{pmatrix}
g_1(x)\\
g_2(x)
\end{pmatrix}=\frac{1}{(r^2+a^2)(r^2+b^2)}\begin{pmatrix}
\Delta(\frac{1}{2r}-i\omega) & -\frac{1}{r} \sgn(\Delta) \sqrt{\abs{\Delta}}S(r)e^{2i\Omega_+ x}\\
-\frac{1}{r}\sqrt{\abs{\Delta}}\;\overline{S}(r)e^{-2i\Omega_+ x} & T(r)
\end{pmatrix},
\end{equation*}
where, using the form of $\Omega_+$ (\ref{omegaplus}),
\begin{equation*}
T(r)=\frac{\Delta}{2r}+2iU(r)+i\omega(\Delta+4Mr^2) - 2i\Omega_+(r^2+a^2)(r^2+b^2)=\frac{\Delta}{2r}-i\omega\Delta.
\end{equation*}
As $r\searrow r_+$ and $x\to\infty$, $\Delta$ decays to zero exponentially, and thus $g_1(x)$ and $g_2(x)$ tend to constants. The exponential decay of the error term $R_+$ then follows from the same proof as in \cite{Finster_2003}.
\end{proof}

The asymptotic behaviour near the Cauchy horizon $r_-$ is similar to the above, as $\Delta\to 0$ exponentially as $r\to r_\pm$. More precisely, we have

\begin{pro}\label{Cauchyasympt}
Every nontrivial solution $\tilde{X}$ of (\ref{radial}) is asymptotically as $r\searrow r_-$ of the form
\begin{equation}
\tilde{X}(x)=\begin{pmatrix}
g_{-,1}\\
g_{-,2}e^{2i\Omega_- x}
\end{pmatrix} + R_-(x)
\end{equation}
with
\begin{align*}
\abs{g_{-,1}}^2 + \abs{g_{-,2}}^2&\neq 0\\
\Omega_- &= \omega + \frac{a(k_a+\frac{1}{2})}{r_-^2+a^2} + \frac{b(k_b+\frac{1}{2})}{r_-^2+b^2}\\
\abs{R_-}&\leq c_- e^{d_- x}
\end{align*}
for suitable constants $c_-,d_->0$.
\end{pro}

To study the asymptotics at infinity, we note that all terms in (\ref{radial}) are analytic in $r$ away from the horizons. As such, we may expand (\ref{radial}) in powers of $1/r$. We obtain
\begin{equation}\label{xtildaODE}
\partial_r \tilde{X} = \left[ \begin{pmatrix}
-i\omega & im \\
-im & i\omega \end{pmatrix}
+ \frac{1}{r} \begin{pmatrix}
\frac{1}{2} & -\lambda \\
-\lambda & \frac{1}{2}
\end{pmatrix} + \mathcal{O}(r^{-2})\right]\tilde{X}.
\end{equation}
The eigenvalues of the matrix potential between the square brackets are not purely imaginary, and the method used by \cite{Finster_2003} cannot be directly applied to (\ref{xtildaODE}). However, this issue may be circumvented by the substitution $\overline{X}_\pm = r^{-1/2}\tilde{X}_\pm$. In terms of the original functions $X_\pm$ from the separation ansatz (\ref{ansatz}), we have 
\begin{equation*}
\overline{X}_- = \sqrt{r\abs{\Delta}}X_- \qquad \mathrm{and} \qquad  \overline{X}_+=\frac{r_+^3}{\sqrt{r}}X_+.
\end{equation*} 
The asymptotic form of the radial ODE at infinity then becomes
\begin{equation}\label{overlineasymptotics}
\partial_r \overline{X} = \left[ \begin{pmatrix}
-i\omega & im \\
-im & i\omega \end{pmatrix}
+ \frac{1}{r} \begin{pmatrix}
0 & -\lambda \\
-\lambda & 0
\end{pmatrix} + \mathcal{O}(r^{-2})\right]\overline{X}.
\end{equation}
Noting that (\ref{overlineasymptotics}) is precisely the same as in \cite{BLpaper}, which is a strict simplification from \cite{Finster_2003}, we therefore have
\begin{pro}\label{infinityasympt}
Every nontrivial solution $\overline{X}$ of (\ref{radial}) has for large $r$ the asymptotic form
\begin{equation}\label{asymptoticsinfty}
\overline{X}(r)=A\begin{pmatrix}
e^{-i\Phi(r)}f_{\infty,1}\\
e^{i\Phi(r)}f_{\infty,2}
\end{pmatrix}
+R_\infty(r),
\end{equation}
with, for some constant $C>0$,
\begin{align*}
\abs{f_{\infty,1}}^2+\abs{f_{\infty,2}}^2&\neq 0\\
\Phi(r)&=\mathrm{sgn}(\omega)\sqrt{\omega^2-m^2}\:r\\
A&=\begin{pmatrix}
\cosh\Theta & \sinh\Theta\\
\sinh\Theta & \cosh\Theta
\end{pmatrix}, \qquad \Theta=\frac{1}{4}\log\left(\frac{\omega+m}{\omega-m}\right)\\
\abs{R_\infty}&\leq \frac{C}{r}.
\end{align*}
\end{pro}

Thus, up to some suitable reparametrization $\overline{X}$, the asymptotic behaviour at infinity of Dirac spinors is identical in Eddington--Finkelstein-type coordinates and Boyer--Lindquist coordinates. This is to be expected, as the two coordinate systems tend towards each other at radial infinity.

\section{Essential self-adjointness of the Dirac Hamiltonian}\label{self-adj}
Across the event and Cauchy horizons of the 5D Myers--Perry black hole, the Dirac Hamiltonian loses ellipticity. Standard results for elliptic operators can therefore not be applied to conclude its essential self-adjointness.
In \cite{Finster_2016}, Finster and Röken construct an essentially self-adjoint extension of the Dirac Hamiltonian for a class of non-uniformly elliptic boundary value problems which include the Kerr geometry \cite{Finster_2018} and the 5D Myers--Perry geometry in Eddington--Finkelstein-type coordinates.
In this section, we apply the main result of \cite{Finster_2016} to obtain essential self-adjointness for the Dirac Hamiltonian with suitable boundary conditions.

Given a system of coordinates $x^\mu=(t,x^j)$ on a spin manifold $\mathcal{M}$ with metric tensor $g=g_{\mu\nu}\, dx^\mu\otimes dx^\nu$, a set of gamma matrices $\gamma^A$ satisfying (\ref{gammaanticom}) and an orthonormal frame $\partial_A={e_A}^\mu \partial_\mu$, one can form Dirac matrices $\tilde{\gamma}^\mu = \gamma^A {e_A}^\mu$. The orthonormal frame condition $\eta^{AB}{e_A}^\mu {e_B}^\nu = g^{\mu\nu}$ implies that they satisfy the anticommutation relations
\begin{equation}\label{metricanticomm}
\{\tilde{\gamma}^\mu,\tilde{\gamma}^\nu\}=2g^{\mu\nu}.
\end{equation}
We can therefore write the Dirac equation (\ref{dirac}) equivalently as
\begin{equation*}
\left(i\tilde{\gamma}^\mu\nabla_\mu - m\right)\psi = 0,
\end{equation*}
where $\nabla$ is the spinor connection. It can be rewritten in the Hamiltonian form
\begin{equation}\label{ham}
i\partial_t \psi = H\psi,
\end{equation}
with
\begin{equation*}
H=-(\tilde{\gamma}^t)^{-1}(i\gamma^j\nabla_j - m).
\end{equation*}
In particular, the principal symbol of $H$ as a partial differential operator is
\begin{equation*}
P(\boldsymbol{\zeta})=\alpha^j\zeta_j\qquad \mathrm{with}\qquad \alpha^j = -i (\tilde{\gamma}^t)^{-1}\tilde{\gamma}^j.
\end{equation*}
Under an invertible, time-independent transformation of spinors $\psi'=\mathscr{P}\psi$, the transformed Hamiltonian form of the Dirac equation becomes $i\partial_t\psi'=H'\psi'$, where the principal symbol of $H'$ is now 
\begin{equation*}
P'(\boldsymbol{\zeta})=\alpha'^j\zeta_j\qquad \mathrm{with}\qquad \alpha'^j= -i\mathscr{P}(\tilde{\gamma}^t)^{-1}\tilde{\gamma}^j\mathscr{P}^{-1}.
\end{equation*}
Its determinant is given by
\begin{equation*}
\det P'(\boldsymbol{\zeta})=\frac{\det(\tilde{\gamma}^j\zeta_j)}{\det(\tilde{\gamma}^t)}.
\end{equation*}
Using the relations $(\tilde{\gamma}^\tau)^2=g^{\tau\tau}$ and $\tilde{\gamma}^j\zeta_j\tilde{\gamma}^k\zeta_k=g^{jk}\zeta_j\zeta_k$ arising from the anticommutation relations (\ref{metricanticomm}), one obtains
\begin{equation*}
\det P'(\boldsymbol{\zeta})=\sqrt{\frac{g^{jk}\zeta_j\zeta_k}{g^{\tau\tau}}}.
\end{equation*}
Thus, the determinant of the principal symbol vanishes and the Dirac Hamiltonian loses ellipticity precisely where the spatial part of the background metric is no longer positive or negative definite.

Now, we consider the 5D Myers--Perry geometry in Eddington--Finkelstein-type coordinates $x^\mu=(\tau,x^j)$, where $x^j=(r,\theta,\phi,\xi)$, with the orthonormal frame (\ref{framevectors}) and the invertible spinor transformation (\ref{Ptransform}). We note that the transformed Dirac Hamiltonian $H'$ is exactly $H'=N\mathcal{D}_0$ of (\ref{hamform}). The components of the spatial part of the inverse of the metric (\ref{EFmetric}) are
\begin{equation*}
g^{jk}=\frac{1}{\Sigma r^2}\begin{pmatrix}
\Delta & 0 & a(r^2+b^2) & b(r^2+a^2) \\
0 & r^2 & 0 & 0\\
a(r^2+b^2) & 0 & r^2\csc^2\theta + b^2 & ab\\
b(r^2+a^2) & 0 & ab & r^2\sec^2\theta + a^2
\end{pmatrix}.
\end{equation*}
At the horizons, where $\Delta=0$, we may take $\boldsymbol{\zeta}=(1,0,0,0)$, which yields $g^{jk}\zeta_j\zeta_k=0$. Thus $H'$ is not elliptic there, and we have to apply the results of \cite{Finster_2016}.

Let $\mathcal{M}=\mathbb{R}_\tau \times [r_0,\infty)_r \times (0,\frac{\pi}{2})_\theta \times (0,2\pi)_\phi \times (0,2\pi)_\xi$ with $r_0<r_-$. Equipped with the metric (\ref{EFmetric}), it is a submanifold with boundary of the 5D Myers--Perry geometry. It has an inner boundary
\begin{equation*}
\partial \mathcal{M}=\{\tau,r=r_0,\theta,\phi,\xi\}
\end{equation*}
and a family of spacelike hypersurfaces
\begin{equation*}
\Sigma_\tau = \{\tau=const.,r,\theta,\phi,\xi\}
\end{equation*}
with compact boundaries
\begin{equation*}
\partial \Sigma_\tau = \{\tau=const., r=r_0, \theta,\phi,\xi\}\simeq S^3.
\end{equation*}

The construction of \cite{Finster_2016} is crucially dependent on the existence of a Killing vector field $K$ which is timelike and tangential at the boundary $\partial\mathcal{M}$, as the choice of time coordinate $t$ used to formulate the Cauchy problem for the Dirac equation is precisely along $K=\partial_t$. However, as in the Kerr geometry, the Killing vector field $\partial_\tau$ is not everywhere timelike on $\partial \mathcal{M}$, as can be seen as follows. The condition for $\partial_\tau$ to be timelike on $\partial\mathcal{M}$, in our chosen metric signature, is
\begin{equation*}
g(\partial_\tau,\partial_\tau)=g_{\tau\tau}=-1+\frac{\mu}{\Sigma}<0,
\end{equation*}
or equivalently
\begin{equation*}
\mu-r_0^2 < a^2\cos^2\theta + b^2\sin^2\theta
\end{equation*}
for all $\theta\in(0,\pi/2)$. Supposing this were satisfied and using the fact that 
\begin{equation*}
r_0^2<r_-^2=\frac{1}{2}(\mu-a^2-b^2)-\sqrt{(\mu-a^2-b^2)^2-4a^2b^2},
\end{equation*}
we then must have
\begin{equation*}
\sqrt{(\mu-a^2-b^2)^2-4a^2b^2}<a^2\left(\cos^2\theta-\frac{1}{2}\right) + b^2\left( \sin^2\theta - \frac{1}{2}\right) - \frac{\mu}{2}
\end{equation*}
for all $\theta\in(0,\pi/2)$. This is clearly false, as the right-hand side simplifies to $-\mu/2<0$ when $\theta=\pi/4$. Therefore the vector field $\partial_\tau$ is not timelike everywhere on $\partial\mathcal{M}$.

Analogously to what was done in \cite{Finster_2018}, it is possible to construct a Killing vector field timelike and tangential at $r=r_0$ in the following way, recalling that $\partial_\tau, \partial_\phi$ and $\partial_\xi$ are commuting Killing vector fields on $\mathcal{M}$.

\begin{lem}
There exist $\alpha,\beta \in \mathbb{R}$ such that $K=\partial_\tau + \alpha \partial_\phi + \beta \partial_\xi$ is tangential and timelike at $\partial \mathcal{M}$.
\end{lem}

\begin{proof}
Clearly $K$ is tangential to $\partial \mathcal{M}$. The condition for $K$ to be timelike is equivalent to
\begin{equation*}
p(\alpha,\beta)\equiv g(K,K)= g_{\tau\tau} + 2\alpha g_{\tau\phi} + 2\beta g_{\tau\xi} + 2\alpha\beta g_{\phi\xi} + \alpha^2 g_{\phi\phi} + \beta^2 g_{\xi\xi}<0.
\end{equation*}
Denote by $\disc_x(q(x))$ the discriminant of a polynomial $q(x)$. For a given $\alpha\in\mathbb{R}$ fixed, since 
\begin{equation*}
g_{\xi\xi}=(r^2+b^2)\cos^2\theta + \frac{\mu b^2\cos^4\theta}{\Sigma}>0,
\end{equation*}
\noindent there exists a $\beta\in\mathbb{R}$ such that $p(\alpha,\beta)>0$ if and only if $\disc_\beta(p(\alpha,\beta))>0$. As $\disc_\beta(p(\alpha,\beta))$ is a degree 2 polynomial in $\alpha$, there is an $\alpha\in\mathbb{R}$ such that $\disc_\beta(p(\alpha,\beta))>0$ if $\disc_\alpha(\disc_\beta(p(\alpha,\beta)))>0$. Doing the explicit computation gives
\begin{equation*}
\disc_\alpha(\disc_\beta(p(\alpha,\beta))) = \frac{64\Delta\cos^6\theta\sin^2\theta}{\Sigma}\left((r^2+b^2)(r^2\sec^2\theta + a^2 + b^2\tan^2\theta) + \mu b^2\right).
\end{equation*}
This is strictly positive at $r=r_0$ for all $\theta\in (0,\pi/2)$, as $\Delta>0$ inside the Cauchy horizon.

\end{proof}

In the Kerr geometry \cite{Finster_2016}, the time coordinate $\tau'$ is chosen along the integral curves of the Killing vector field $K$, i.e. $K=\partial_{\tau'}$, then ``spatial'' coordinates are chosen on constant $\tau'$ hypersurfaces. Analogously, in the 5D Myers-Perry geometry we may relate the coordinate system $(\tau',r,\theta,\phi',\xi')$ to the Eddington--Finkelstein-type coordinates $(\tau,r,\theta,\phi,\xi)$ by the coordinate transformation
\begin{equation*}
\tau'=\tau, \qquad \phi'=\phi-\alpha\tau, \qquad \xi'=\xi-\beta \tau.
\end{equation*}
The induced metric on constant $\tau'$ hypersurfaces is the same as that on constant $\tau$ hypersurfaces, since we have $d\tau'=d\tau=0$, which results in $d\phi=d\phi'$ and $d\xi=d\xi'$. Furthermore, $\tau'$ is a proper time function just as the original time function $\tau$, as can be seen by the following consideration. The gradient 
\begin{equation*}
\nabla \tau'=g^{\mu'\nu'}(\partial_{\mu '} \tau') \partial_{\nu'} = g^{\tau'\nu'}\partial_{\nu'}
\end{equation*}
satisfies 
\begin{equation*}
g(\partial_{\tau'},\nabla \tau') = g^{\tau'}_{\tau'} = 1>0 \qquad \textrm{and}\qquad g(\nabla \tau', \nabla \tau')=g^{\tau'\tau'}=g^{\tau\tau}=-1-\frac{\mu}{\Sigma}<0,
\end{equation*}
showing that $\nabla \tau'$ is future-pointing and timelike. As such, the results of \cite{Finster_2016}, valid for the coordinates $(\tau',r,\theta,\phi',\xi')$, will also hold in Eddington--Finkelstein-type coordinates $(\tau,r,\theta,\phi,\xi)$.

The untransformed Dirac Hamiltonian $H$ is symmetric  \cite{Finster_2016} with respect to the scalar product
\begin{equation*}
(\psi|\phi)=\int_{\Sigma_\tau}\!\!\prec\!\! \psi | \slashed{\nu} \phi \!\!\succ d\mu
\end{equation*}
together with the boundary condition
\begin{equation*}
(\slashed{n}-i)\psi|_{\partial \Sigma_\tau}=0,
\end{equation*}
where $\mu$ is the induced measure on constant $\tau$ hypersurfaces $\Sigma_\tau$, $\nu$ is the future-directed, timelike normal, $n$ is the inner normal on $\partial\mathcal{M}$, the slashes are Clifford contraction with respect to $\tilde{\gamma}^\mu$ and
\begin{equation*}
\prec \cdot \;|\; \cdot \succ\; : S_p\mathcal{M} \times S_p\mathcal{M} \rightarrow \mathbb{C}
\end{equation*}
is the spin scalar product of signature $(2,2)$, provided by the spin structure on $\mathcal{M}$. The scalar product is independent of the choice of $\tau$.

Define the scalar product $\langle\cdot|\cdot\rangle$ on the transformed spinors $\psi',\phi'$ by
\begin{equation*}
\langle \psi' | \phi' \rangle = (\mathscr{P}^{-1}\psi'|\mathscr{P}^{-1}\phi'),
\end{equation*}
with $\mathscr{P}$ given by (\ref{Ptransform}). Note that $\langle \cdot | \cdot \rangle$ is bounded near the boundary $\partial\Sigma_\tau$, for both $\slashed{\nu}=\gamma^\alpha \nu_\alpha$ and $\mathscr{P}$ are bounded near the boundary. Let $\mathcal{H}$ and $\mathcal{G}$ be the Hilbert spaces obtained by completing the space of spinors with respect to $(\cdot|\cdot)$ and $\langle\cdot|\cdot\rangle$ respectively. By construction, the spinor transformation $\psi \mapsto \psi'=\mathscr{P}\psi$ is a Hilbert space isometry $\mathcal{H} \rightarrow \mathcal{G}$. Since $\mathscr{P}$ is time-independent, the Dirac equation in Hamiltonian form (\ref{ham}) is equivalent to
\begin{equation*}
\left(i\partial_\tau - \mathscr{P}H\mathscr{P}^{-1}\right)\psi'=0,
\end{equation*}
and therefore $H'=\mathscr{P}H\mathscr{P}^{-1}$. Then, clearly $H'$ is symmetric with respect to $\langle\cdot|\cdot\rangle$ on spinors that vanish on $\partial \mathcal{M}$, for
\begin{equation*}
\langle\psi'|H'\phi'\rangle = \langle \phi'|\mathscr{P}H\mathscr{P}^{-1}\phi'\rangle = (\mathscr{P}^{-1}\psi'|H\mathscr{P}^{-1}\phi') = (H\mathscr{P}^{-1} \psi' | \mathscr{P}^{-1}\phi') = \langle H'\psi'|\phi'\rangle.
\end{equation*}
For spinors that may be nonzero on $\partial\mathcal{M}$, the boundary condition that one must impose to preserve the symmetry of the Hamiltonian are such that $\prec\!\!\psi'|\slashed{n}\phi'\!\!\succ$ vanishes on $\partial \Sigma_\tau$ \cite{Finster_2016,Finster_2018}. As in the Kerr geometry \cite{Finster_2018}, it is sufficient to require that
\begin{equation*}
\left(\slashed{n}-i\mathscr{P}(\mathscr{P}^\dag)^{-1}\right)\psi'|_{\partial \mathcal{M}}=0.
\end{equation*}
We therefore can apply \cite[Theorem 1.2]{Finster_2016} and conclude that

\begin{thm}
The Dirac Hamiltonian $H'$ in the non-extreme 5D Myers--Perry geometry in Eddington--Finkelstein-type coordinates with domain of definition
\begin{equation}\label{domain}
\mathrm{Dom}(H')=\{\psi'\in C_0^\infty(\Sigma_\tau)^4|\left(\slashed{n}-i\mathscr{P}(\mathscr{P}^\dag)^{-1}\right)(H'^p\psi')|_{\partial \Sigma_\tau}=0 \text{ for each integer } p\geq 0\}
\end{equation}
is essentially self-adjoint.
\end{thm}

Note that $\mathrm{Dom}(H')$ is dense in $L^2(\Sigma_\tau,\langle \cdot | \cdot \rangle)^4$, for it is a subset of $C_0^\infty(\Sigma_\tau)^4$ with boundary conditions on a compact boundary $\partial\Sigma_\tau\cong S^3$, and the scalar product $\langle\cdot|\cdot\rangle$ is bounded near $\partial\Sigma_\tau$.

\section{Angular spectral projectors}\label{angularprojectors}

In order to write the resolvents of the Dirac Hamiltonian $H'$ in terms of the Green's matrix of the radial ODE (\ref{radial}), we need to decompose the Hilbert space of spinors into eigenspaces of the angular Dirac operator given in (\ref{angular}). The main difficulty in doing so is that in order to use Stone's formula, one must consider the angular operator $A=A(\omega_\epsilon)$ with slightly complex frequencies $\omega_\epsilon \equiv \omega \pm i\epsilon$, which is no longer self-adjoint. In this section, we use the method of slightly non-self-adjoint perturbations \cite{Finster_2006} to construct spectral projectors $\{Q^\pm_l\}_{ l\in \mathbb{Z}}$ onto one-dimensional eigenspaces of $A(\omega_\epsilon)$ for small masses $|m|$ and suitably bounded frequencies $|\omega|$. To do so, we use the following result from \cite{Finster_2006}, adapted for our purposes.

\begin{thm}\label{slightlynonselfadjoint}
    Suppose that $T$ is a self-adjoint operator with nondegenerate spectrum $\{\lambda_n\}_{n\in\mathbb{Z}}$ and that there exists an $\epsilon>0$ such that $\lambda_{n+1}-\lambda_n>\epsilon$ for all $n$. Suppose also that $W$ is a bounded operator satisfying $\lVert W\rVert_{op}\leq \epsilon/16$ and let $A=T+W$. Then there exists a family of spectral projectors $\{Q_l\}_{l\in\mathbb{Z}}$ such that
    \begin{itemize}
        \item The image of $Q_l$ is a 1-dimensional eigenspace of $A$ for each $l$.
        \item $Q_l$ is a projection for each $l$, that is, $Q_l^2=Q_l$.
        \item $\{Q_l\}_{l\in\mathbb{Z}}$ is complete, that is, $\sum_{l} Q_l = \mathds{1}$ with strong convergence of the series.
    \end{itemize}
\end{thm}

In \cite[Section 8]{Finster_2006}, Theorem \ref{slightlynonselfadjoint} is shown for an unperturbed operator $T$ with positive spectrum $0\leq \lambda_1 < \lambda_2 < \cdots$. Nonetheless, the same proof works for a two-sided spectrum $\cdots < \lambda_{-1}<\lambda_0<\lambda_1<\cdots$.

The angular operator $A$ of (\ref{angular}) acts on $L^2((0,\pi/2),\sin\theta\cos\theta\,d\theta)^2$. It may be written as
\begin{equation*}
    A(\omega_\epsilon) = \mathcal{D}^{k_ak_b}_{S^3} +mp\sigma_3 +\omega V \pm i\epsilon V,
\end{equation*}
where $L_\theta$ is as in (\ref{angularabbrevs}), $\sigma_3$ is one of the Pauli matrices shown in (\ref{paulimatrices}),
\begin{equation}\label{DS3biazimuthal}
    \mathcal{D}^{k_ak_b}_{S^3} = \begin{pmatrix}
        -\frac{k_b+1/2}{\cos\theta} & L_\theta - \frac{k_a+1/2}{\sin\theta}\\
        -L_\theta - \frac{k_a+1/2}{\sin\theta} & \frac{k_b+1/2}{\cos\theta}
        \end{pmatrix},\quad V= \begin{pmatrix}
            -\frac{ab}{p} & -\frac{(a^2-b^2)\sin\theta\cos\theta}{p} \\
            -\frac{(a^2-b^2)\sin\theta\cos\theta}{p} & \frac{ab}{p}
        \end{pmatrix}.
\end{equation}
To apply Theorem \ref{slightlynonselfadjoint}, we will consider the unperturbed operator to be $A(\omega)=\mathcal{D}^{k_ak_b}_{S^3} + mp\sigma_3 + \omega V$, and denote its eigenvalues by $\{\lambda_l\}_{l\in\mathbb{Z}}$. By Proposition \ref{angularnondegenerate}, it suffices to require that consecutive eigenvalues of $A(\omega)$ be separated by a uniform nonzero gap $\lambda_{l+1}-\lambda_l>\epsilon>0$ for all $l\in \mathbb{Z}$. For general $m$ and $\omega$, it is not clear how to obtain such a spectral gap.

The key is to notice that $\mathcal{D}^{k_ak_b}_{S^3}$ is equal to the intrinsic Dirac operator on $S^3$ with the usual metric and a particular choice of representation for the Clifford algebra, restricted to the bi-azimuthal mode $(k_a,k_b)$. More precisely, the Dirac operator on $S^3$ in Hopf coordinates $(\theta,\varphi,\psi)$ is given \cite{Daud__2012} by (\ref{DS3}). Choosing the 3-dimensional representation $\Gamma^2=\sigma_2,\; \Gamma^3 = -\sigma_1,\; \Gamma^5 = -\sigma_3$, where $\sigma_j$ are the Pauli matrices (\ref{paulimatrices}), and taking $\partial_\varphi\to -i(k_a+1/2)$ and $\partial_\psi\to -i(k_b+1/2)$, we see that (\ref{DS3}) agrees exactly with (\ref{DS3biazimuthal}), as required.

The spectrum of the Dirac operator on $S^3$, acting on $L^2(S^3)^2$, is known \cite{bar} to be $\sigma(\mathcal{D}_{S^3})=\{\pm \frac{3}{2}, \pm \frac{5}{2}, \pm \frac{7}{2},\dots\}$, and as a result $\sigma(\mathcal{D}^{k_ak_b}_{S^3})\subset \{\pm \frac{3}{2}, \pm \frac{5}{2}, \pm \frac{7}{2},\dots\}$. In particular, any two eigenvalues of $\mathcal{D}^{k_ak_b}_{S^3}$ differ by at least $1$ in absolute value. Considering $A(\omega)$ as a bounded perturbation of $\mathcal{D}^{k_ak_b}_{S^3}$ and denoting by $\lVert\cdot\rVert_{op}$ the operator norm on $L^2((0,\pi/2), \sin\theta\cos\theta\,d\theta)^2$ and by $\lVert\cdot\rVert$ the $\theta$-pointwise matrix norm, we estimate
\begin{equation}\label{perturbationestimate}
    \lVert A(\omega)-\mathcal{D}^{k_ak_b}_{S^3}\rVert_{op} = \lVert mp\sigma_3 + \omega V\rVert_{op} \leq \sup_\theta \lVert mp(\theta)\sigma_3 + \omega V(\theta)\rVert \leq \sup_{\theta}\left(|m||p(\theta)|+|\omega|\lVert V(\theta)\rVert\right).
\end{equation}
In (\ref{perturbationestimate}), we used the fact that for a zeroth order matrix operator $T(\theta)$ on $L^2$, $\lVert T \rVert_{op}\leq \sup_\theta \lVert T(\theta)\rVert$. An explicit computation of the eigenvalues of $V(\theta)$ yields, recalling that $p=\sqrt{a^2\cos^2\theta + b^2\sin^2\theta}$,
\begin{align*}
    \lVert V(\theta)\rVert^2 &= \frac{a^2b^2}{p^2} + \frac{(a^2-b^2)^2\sin^2\theta\cos^2\theta}{p^2}\\
    &= \frac{a^2b^2(\cos^4\theta + \sin^4\theta)+(a^4+b^4)\sin^2\theta\cos^2\theta}{a^2\cos^2\theta + b^2\sin^2\theta}\\
    &= a^2\sin^2\theta + b^2\cos^2\theta.
\end{align*}

As the spectral gaps of $\mathcal{D}^{k_ak_b}_{S^3}$ are of size 1, in order to obtain a nonzero uniform gap between each pair of eigenvalues of $A(\omega)$, it suffices by standard perturbation theory \cite[Theorem 4.10]{kato} to require that $\lVert A(\omega)-\mathcal{D}^{k_ak_b}_{S^3}\rVert_{op}<\frac{1}{2}$. In terms of the above estimates, the latter bound is implied by
\begin{equation*}
    \sup_{\theta}\left( |m|\sqrt{a^2\cos^2\theta+b^2\sin^2\theta} + |\omega|\sqrt{a^2\sin^2\theta + b^2\cos^2\theta}\right)<\frac{1}{2},
\end{equation*}
which in turn is implied by
\begin{equation*}
	|m|+|\omega|< \frac{1}{2\max(|a|,|b|)}.
\end{equation*}
In particular, we must have 
\begin{equation}\label{mbound}
    |m|< \frac{1}{2\max(|a|,|b|)}.
\end{equation}
Then, it would suffice to require that
\begin{equation}\label{omegabound}
    |\omega|< \frac{1}{2\max(|a|,|b|)}-|m|.
\end{equation}
In which case, by \cite[Theorem 4.10]{kato}, the consecutive eigenvalues of $A(\omega)$ would satisfy
\begin{equation}\label{spectralgap}
    \lambda_{l+1}-\lambda_l \geq 1- 2\lVert A(\omega)-\mathcal{D}^{k_ak_b}_{S^3}\rVert_{op} \geq 1-2(|m|+|\omega|)\max (|a|,|b|).
\end{equation}
Notice that (\ref{omegabound}) implies that the right-hand side of (\ref{spectralgap}) is strictly positive. In order to satisfy the hypothesis of Theorem \ref{slightlynonselfadjoint} on the norm of the non-self-adjoint perturbation, it suffices then to require that for all $l\in \mathbb{Z}$,
\begin{equation*}
    \lVert \pm i\epsilon V\rVert_{op} \leq \epsilon \sup_\theta \sqrt{a^2\sin^2\theta+b^2\cos^2\theta}\leq \epsilon \max(|a|,|b|) < \frac{1}{16}(\lambda_{l+1}-\lambda_l).
\end{equation*}
This is implied by
\begin{equation*}
    \epsilon <  \frac{1}{16\max(|a|,|b|)} - \frac{|m|+|\omega|}{8}.
\end{equation*}
We have therefore shown the following proposition.
\begin{pro}\label{Qlprop}
    If $|m|< \frac{1}{2\max(|a|,|b|)}$ and $0\leq E_0<\frac{1}{2\max(|a|,|b|)}-|m|$, then there exists an $\epsilon>0$ such that whenever $|\omega|\leq E_0$, there exists a complete set of spectral projectors $\{Q^\pm_l\}_{l\in \mathbb{Z}}$ onto the 1-dimensional eigenspaces of $A(\omega_\epsilon) = A(\omega \pm i\epsilon)$.
\end{pro}

\section{Integral spectral representation}\label{intrep}

With the essential self-adjointness of $H'$ and the existence of angular projectors $Q^\pm_l$ established, we are now ready to construct in this section an integral spectral representation for the Dirac propagator through the horizons of the 5D Myers--Perry geometry for initial data with small $m$ and frequency spectrum contained in an interval $(-E_0, E_0)$ with $E_0$ as in Proposition \ref{Qlprop}. More specifically, we will use a variant of Stone's formula \cite{kronthaler}, which states in our context that
\begin{multline}\label{stone}
    \frac{1}{2}e^{-i\tau H'_{k_ak_b}}(P_{[-E_0,E_0]}+P_{(-E_0,E_0)})\psi'_{0,k_a,k_b}\\
    =\frac{1}{2\pi i}\lim_{\epsilon\searrow 0} \int_{-E_0}^{E_0} e^{-i\omega\tau} \left[ (H'_{k_ak_b}-\omega-i\epsilon)^{-1}-(H'_{k_ak_b}-\omega+i\epsilon)^{-1}\right] \psi'_{0,k_a,k_b}\, d\omega.
\end{multline}
The requirement we impose on the frequency spectrum of the initial data is most conveniently written as
\begin{equation}\label{lowenergy}
    \psi'_0\in \mathrm{im}(P_{(-E_0,E_0)}),
\end{equation}
where $P$ is the spectral projector for the Dirac Hamiltonian $H'$. In this case, we have
\begin{equation}\label{lowenergyidentity}
    \frac{1}{2}(P_{[-E_0,E_0]}+P_{(-E_0,E_0)})\psi'_0 = \psi'_0.
\end{equation}
As the Dirac Hamiltonian commutes with $\partial_\phi$ and $\partial_\xi$, (\ref{lowenergyidentity}) also holds for each bi-azimuthal mode, that is,
\begin{equation}\label{lowenergybiazimuthal}
\frac{1}{2}(P_{[-E_0,E_0]}+P_{(-E_0,E_0)})\psi'_{0,k_a,k_b} = \psi'_{0,k_a,k_b} \quad \mathrm{for\: each} \quad k_a,k_b\in \mathbb{Z}.
\end{equation}
Before stating the main result of this section, it is convenient to recast the radial ODE (\ref{radial}) as 
\begin{equation*}
\mathcal{R}^{2\times 2} \begin{pmatrix}
X_+\\
X_-
\end{pmatrix}=0,
\end{equation*}
where
\begin{equation}\label{R22}
\mathcal{R}^{2\times 2}=\begin{pmatrix}
r_+^3\mathcal{R}_+ & \overline{S}_l\\
S_l & r_+^{-3}\mathcal{R}_-
\end{pmatrix},
\end{equation}
with
\begin{align*}
\mathcal{R}_+&=\frac{1}{r^2}\partial_r - \frac{1}{2r^3} + \frac{i\omega}{r^2}\\
\mathcal{R}_-&= \Delta \partial_r - 2i U(r) + \frac{\partial_r\Delta}{2} + \frac{\Delta}{2r} - i\omega(\Delta + 4Mr^2)
\end{align*}
and $U(r)$, $S_l(r)$ as in (\ref{Ur}), (\ref{Sr}) respectively, with the replacement $\lambda\to \lambda^\pm_l$.

\begin{thm}\label{intreptheorem}
Suppose that $\psi(\tau,r,\theta,\phi,\xi)$ satisfies the massive Dirac equation (\ref{dirac}) across the horizons of the non-extreme 5-dimensional Myers--Perry geometry with $\psi_0(\cdot)\equiv \psi(0,\cdot)\in C^\infty_0 ((r_0,\infty)\times S^3)^4$ for some $0<r_0<r_-$ and with $|m|< \frac{1}{2\max(|a|,|b|)}$. Suppose furthermore that there exists an $E_0< \frac{1}{2\max(|a|,|b|)}-|m|$ such that $\psi_0' = \mathscr{P}\psi_0 \in \mathrm{im}(P_{(-E_0,E_0)})$, where $P$ is the spectral projector for the Dirac Hamiltonian $H'$.

Then, the transformed solution $\psi'(\tau,r,\theta,\phi,\xi)$ admits the integral spectral representation
\begin{multline}\label{intrepformula}
\psi'(\tau,r,\theta,\phi,\xi)=\frac{1}{2\pi i}\sum_{k_a,k_b\in\mathbb{Z}} e^{-i((k_a+\frac{1}{2})\phi + (k_b+\frac{1}{2})\xi)} \\
\times \lim_{\epsilon\searrow 0} \int_{-E_0}^{E_0}e^{-i\omega\tau} \left((H'_{k_ak_b} - \omega - i\epsilon)^{-1} - (H'_{k_ak_b} - \omega + i\epsilon)^{-1}\right)(r,\theta;r',\theta')\psi'_{0,k_a,k_b}(r',\theta')\, d\omega,
\end{multline}
where $\psi'_{0,k_a,k_b}\in C_0^\infty\left( (r_0,\infty)\times [0,\pi/2]\right)^4$ is the initial data for fixed $k_a,k_b$ and $(H'_{k_ak_b}-\omega\mp i\epsilon)^{-1}$ are the resolvents of the Dirac Hamiltonian on the upper and lower complex half-planes. The resolvents are unique and of the form
\begin{multline*}
(H'_{k_ak_b}-\omega\mp i\epsilon)^{-1}(r,\theta;r',\theta')\psi'_{0,k_a,k_b}(r',\theta')\\
=-\sum_{l\in\mathbb{Z}}\int_{-1}^1 Q^\pm_l(\theta;\theta')\int_{r_0}^\infty \mathscr{C} \begin{pmatrix}
G(r;r')_{k_a,k_b,\omega\pm i\epsilon} & \mathbf{0}\\
\mathbf{0} & G(r;r')_{k_a,k_b,\omega\pm i\epsilon}
\end{pmatrix} \mathscr{E}(r',\theta')\psi'_{0,k_a,k_b}(r',\theta')\, dr'\, d\left(\cos\theta'\right)
\end{multline*}
with $Q^\pm_l(\cdot,\cdot)$ the integral kernel of the spectral projector onto a 1-dimensional, invariant subspace of the angular operator (\ref{angular}) corresponding to the angular eigenvalue $\lambda^\pm_l$ and the frequency $\omega\pm i\epsilon$, $G(r;r')_{k_a,k_b,\omega\pm i\epsilon}$ the two-dimensional Green matrix of the radial first order ODE (\ref{radial}), 
\begin{equation*}
\mathscr{C}=\begin{pmatrix}
1 & 0 & 0 & 0\\
0 & 0 & 0 & 1\\
0 & 1 & 0 & 0\\
0 & 0 & 1 & 0
\end{pmatrix},
\end{equation*}
and
\begin{equation*}
\mathscr{E}(r,\theta)=\begin{pmatrix}
\frac{ir_+^3}{r^2} & -\frac{\sin\theta\cos\theta(a^2-b^2)}{p} & ab\left(-\frac{i}{r}-\frac{1}{p}\right) & 0\\
ab\left(\frac{i}{r}-\frac{1}{p}\right) & 0 & -\frac{i}{r_+^3}(\Delta+4Mr^2) & -\frac{\sin\theta\cos\theta(a^2-b^2)}{p}\\
0 & ab\left(-\frac{i}{r}+\frac{1}{p}\right) & -\frac{\sin\theta\cos\theta(a^2-b^2)}{p} & \frac{ir_+^3}{r^2}\\
-\frac{\sin\theta\cos\theta(a^2-b^2)}{p} & -\frac{i}{r_+^3}(\Delta+4Mr^2) & 0 & ab\left(\frac{i}{r}+\frac{1}{p}\right)
\end{pmatrix}.
\end{equation*}
For a general $\psi'_0\in C^\infty_0((r_0,\infty)\times S^3)^4$ that may not satisfy (\ref{lowenergy}) but nonetheless must satisfy the mass restriction (\ref{mbound}), the right-hand side of (\ref{intrepformula}) is instead equal to
\begin{equation}\label{highenergyintrep}
    \frac{1}{2}e^{-i\tau H'}(P_{[-E_0,E_0]}+P_{(-E_0,E_0)})\psi'_0(r,\theta,\phi,\xi).
\end{equation}
\end{thm}

Before we begin the proof, we remark that the matrix $\mathscr{E}(r,\theta)$, up to some column swaps and factors of the event horizon radius $r_+$, is similar in form to the analogous matrix $\mathscr{E}_K(r,\theta)$ arising in the corresponding integral spectral representation in the Kerr geometry \cite{Finster_2018}, namely
\begin{equation*}
\mathscr{E}_K(r,\theta)=-\begin{pmatrix}
i(\Delta_K + 4Mr) & r_+a\sin\theta & 0 & 0\\
0 & 0 & -ir_+ & a\sin\theta\\
0 & 0 & r_+a\sin\theta & i(\Delta_K + 4Mr)\\
a\sin\theta & -ir_+ & 0 & 0
\end{pmatrix}.
\end{equation*}

\begin{proof}[Proof of Theorem \ref{intreptheorem}]
Consider a spinor $\psi'$ with complex frequency $\omega_\epsilon$ which is an $i\partial_\phi$ and $i\partial_\xi$ eigenstate, namely 
\begin{equation*}
\psi'=e^{-i\omega_\epsilon \tau + (k_a+\frac{1}{2})\phi + (k_b+\frac{1}{2})\xi} \Psi(r,\theta).
\end{equation*}
The Dirac equation then takes the form
\begin{equation}\label{kakbdirac}
(H'_{k_ak_b}-\omega_\epsilon)\Psi=0.
\end{equation}
Since the hypotheses of Proposition \ref{Qlprop} are satisfied, there exists for each $l\in\mathbb{Z}$ an idempotent angular spectral projector
\begin{equation*}
Q_l^\pm\Psi = \int_{-1}^1Q_l^\pm(\theta;\theta')\Psi(r,\theta')\, d\left(\cos\theta'\right)
\end{equation*}
onto a finite-dimensional, invariant subspace of the angular operator $A(\omega\pm i\epsilon)$ (\ref{angular}) corresponding to the angular eigenvalue $\lambda^\pm_l$. The angular operator may then be written as 
\begin{equation*}
A(\omega \pm i\epsilon)=\sum_{l\in\mathbb{Z}} \lambda^\pm_l Q^\pm_l.
\end{equation*}
Using the relation $H'=N\mathcal{D}_0$, (\ref{kakbdirac}) can be rewritten as 
\begin{equation*}
N(\mathcal{D}_0 - \omega_\epsilon N^{-1})\Psi=\frac{1}{r}N(\mathcal{R}_{\omega_\epsilon} + \mathcal{A}_{\omega_\epsilon})\Psi = 0.
\end{equation*}
On fixed $k_a,k_b$-modes, the angular system (\ref{angular}) is equivalent to $\mathcal{A}\Psi=\sum_{l\in\mathbb{Z}}\gamma^1\lambda^\pm_l Q^\pm_l\Psi$. We therefore have
\begin{equation*}
\frac{1}{r}N\sum_{l\in\mathbb{Z}}(\mathcal{R}_{\omega_\epsilon} + \gamma^1\lambda^\pm_l)Q^\pm_l\Psi = 0.
\end{equation*}
We note that we may write the above in terms of the radial operator  (\ref{R22}) with $\omega\to\omega_\epsilon$, $\lambda\to\lambda^\pm_l$ as
\begin{equation*}
\mathcal{R}+\gamma^1\lambda^\pm_l = \begin{pmatrix}
i \mathcal{R}_+ & 0 & i\overline{S}_l & 0 \\
0 & -i\mathcal{R}_- & 0 & -iS_l\\
-iS_l & 0 & -i\mathcal{R}_- & 0\\
0 & i\overline{S}_l & 0 & i\mathcal{R}_+
\end{pmatrix}=-i\Gamma^1 \mathscr{C} \begin{pmatrix}
\mathcal{R}^{2\times 2} & \mathbf{0} \\
\mathbf{0} & \mathcal{R}^{2\times 2}
\end{pmatrix}\mathscr{C}^{-1},
\end{equation*}
with
\begin{equation*}
\mathscr{C}=\begin{pmatrix}
1 & 0 & 0 & 0\\
0 & 0 & 0 & 1\\
0 & 1 & 0 & 0\\
0 & 0 & 1 & 0
\end{pmatrix}.
\end{equation*}
Defining $\mathscr{E}^{-1}(r,\theta)=\frac{i}{r} N(r,\theta)\Gamma^1\mathscr{C}$, we obtain
\begin{equation*}
-\mathscr{E}^{-1}(r,\theta) \sum_{l\in\mathbb{Z}}\begin{pmatrix}
\mathcal{R}^{2\times 2} & \mathbf{0} \\
\mathbf{0} & \mathcal{R}^{2\times 2}
\end{pmatrix}\mathscr{C}^{-1}Q^\pm_l\Psi=0.
\end{equation*}
Using (\ref{Ninverse}) yields, via an explicit computation,
\begin{multline*}
\mathscr{E}(r,\theta)=-i\mathscr{C}^{-1}\Gamma^1N^{-1}(r,\theta)\\
=-i\mathscr{C}^{-1}\Gamma^1\Biggl[\left(\frac{\Delta}{2r_+^3}+\frac{2Mr^2}{r_+^3}+\frac{r_+^3}{2r^2}\right) + \Gamma^1\left(\frac{\Delta}{2r_+^3}+\frac{2Mr^2}{r_+^3}-\frac{r_+^3}{2r^2}\right) + \Gamma^3 \frac{\sin\theta\cos\theta(a^2-b^2)}{p} + \Gamma^5\frac{ab}{p} + \gamma^0 \frac{iab}{r}\Biggr]\\
=\begin{pmatrix}
\frac{ir_+^3}{r^2} & -\frac{\sin\theta\cos\theta(a^2-b^2)}{p} & ab\left(-\frac{i}{r}-\frac{1}{p}\right) & 0\\
ab\left(\frac{i}{r}-\frac{1}{p}\right) & 0 & -\frac{i}{r_+^3}(\Delta+4Mr^2) & -\frac{\sin\theta\cos\theta(a^2-b^2)}{p}\\
0 & ab\left(-\frac{i}{r}+\frac{1}{p}\right) & -\frac{\sin\theta\cos\theta(a^2-b^2)}{p} & \frac{ir_+^3}{r^2}\\
-\frac{\sin\theta\cos\theta(a^2-b^2)}{p} & -\frac{i}{r_+^3}(\Delta+4Mr^2) & 0 & ab\left(\frac{i}{r}+\frac{1}{p}\right)
\end{pmatrix}.
\end{multline*}
By the standard theory of ODEs \cite{roach_1982}, there exist Green's functions $G(r;r')$ solving for each $\omega_\epsilon, l, k_a, k_b$ the distributional equation 
\begin{equation*}
\mathcal{R}^{2\times 2}(\partial_r;r) G(r;r')=\delta(r-r')\mathds{1}_2,
\end{equation*}
in terms of which the resolvent of the Dirac Hamiltonian is
\begin{equation*}
(H'_{k_ak_b}-\omega_\epsilon)^{-1}\Psi = -\sum_{l\in\mathbb{Z}} Q^\pm_l \int_{r_0}^\infty \mathscr{C} \begin{pmatrix}
G(r;r') & 0 \\
0 & G(r;r')
\end{pmatrix} \mathscr{E}(r',\theta) \Psi(r',\theta)\, dr'.
\end{equation*}
The identity 
\begin{equation*}
(H'_{k_ak_b}-\omega_\epsilon)(H'_{k_ak_b}-\omega_\epsilon)^{-1}\Psi=\Psi
\end{equation*}
may be verified in the exact same manner as in \cite{Finster_2018}. 

To derive the integral spectral representation (\ref{intrepformula}), we expand the smooth, compactly supported initial data $\psi'_0$ in terms of $k_a,k_b$-modes as 
\begin{equation*}
\psi' = e^{-i\tau H'} \psi'_0 = \sum_{k_a,k_b\in\mathbb{Z}} e^{-i((k_a+\frac{1}{2})\phi + (k_b+\frac{1}{2})\xi)}e^{-i\tau H'_{k_ak_b}}\psi'_{0,k_a,k_b}.
\end{equation*}
The form of the integral representation (\ref{intrepformula}) then follows from Stone's formula (\ref{stone}), the fact that identity (\ref{lowenergybiazimuthal}) holds for $\psi_0'\in \mathrm{im}(P_{(-E_0,E_0)})$ and the same considerations as in \cite{Finster_2018}. In the case $\psi_0'\notin \mathrm{im}(P_{(-E_0,E_0)})$, (\ref{highenergyintrep}) is precisely the left-hand side of Stone's formula (\ref{stone}) summed over the bi-azimuthal modes.
\end{proof}

In the future, it would be of interest to prove Proposition \ref{Qlprop} without the restriction $|m|+|\omega|< \frac{1}{2\max (|a|,|b|)}$, using another method than that of \cite{Finster_2006}. This would enable the construction of an integral spectral representation for the full Dirac propagator, which would open the door to the study of the Green's functions $G(r; r')_{k_a,k_b,\omega\pm i\epsilon}$ using Jost equation methods, as in \cite{Finster_2018} and \cite{kronthaler}.

\bigskip

\noindent \textbf{Acknowledgements: }The author thanks Niky Kamran for discussions, guidance and careful proofreading of the manuscript, as well as Felix Finster for helpful comments and discussions. The author also thanks the referee for valuable comments and suggestions. This work was supported by the Natural Sciences and Engineering Research Council (NSERC) Undergraduate Student Research Award (USRA) program and by NSERC grant RGPIN 105490-2018.
\pagebreak

\appendix

\section{Regular orthonormal frame and connection coefficients}\label{frameappendix}

In this appendix, we transform the frame (\ref{wuframe}) defined on the 5D Myers--Perry geometry into a regular orthonormal frame across the event and Cauchy horizons by using Eddington--Finkelstein-type coordinates and the regularising local Lorentz transformation (\ref{regularizing}). We then present explicit formulae for the components of the connection 1-form and the spinor connection in this regular frame, computed using the first Cartan structure equation.

The 5-dimensional Myers--Perry black hole has a frame consisting of a pair of real principal null vectors $\{\mathbf{l},\mathbf{n}\}$, a pair of complex-conjugate principal null vectors $\{\mathbf{m},\mathbf{\overline{m}}\}$, and a spatial vector $\mathbf{k}$. Letting $(\cdot,\cdot)$ denote the metric scalar product, the above frame satisfies $(\mathbf{l},\mathbf{n})=-1$, $(\mathbf{m},\mathbf{\overline{m}})=1$ and all other pairwise products vanish. In Boyer--Lindquist coordinates, it reads \cite{Wu_2008}
\begin{equation}\label{wuframe}
	\begin{split}
		\mathbf{l}&=\frac{(\rho^2+a^2)(\rho^2+b^2)}{\abs{\Delta}}\left( \partial_t + \frac{a}{\rho^2+a^2}\partial_\varphi + \frac{b}{\rho^2+b^2}\partial_\psi \right) + \sgn(\Delta)\partial_\rho\\
		\mathbf{n}&=\sgn(\Delta)\frac{(\rho^2+a^2)(\rho^2+b^2)}{2\rho^2\Sigma}\left( \partial_t + \frac{a}{\rho^2+a^2}\partial_\varphi + \frac{b}{\rho^2+b^2}\partial_\psi \right) - \frac{\abs{\Delta}}{2\rho ^2\Sigma}\partial_\rho\\
		\mathbf{m}&=\frac{1}{\sqrt{2}(\rho +ip)}\left(\partial_\vartheta + i \frac{\sin\vartheta\cos\vartheta}{p}\left((a^2-b^2)\partial_t + \frac{a}{\sin^2\vartheta}\partial_\varphi - \frac{b}{\cos^2\vartheta}\partial_\psi\right)\right)\\
		\mathbf{\overline{m}}&=\frac{1}{\sqrt{2}(\rho -ip)}\left(\partial_\vartheta - i \frac{\sin\vartheta\cos\vartheta}{p}\left((a^2-b^2)\partial_t + \frac{a}{\sin^2\vartheta}\partial_\varphi - \frac{b}{\cos^2\vartheta}\partial_\psi\right)\right)\\
		\mathbf{k}&=\frac{1}{\rho p}(ab\partial_t + b\partial_\varphi + a\partial_\psi),
	\end{split}
\end{equation}
where $p=\sqrt{a^2\cos^2\theta + b^2\sin^2\theta}$. Let $r_+=\rho_+$ denote the event horizon radius. Applying a local Lorentz transformation $\mathbf{l}'=\varsigma \mathbf{l}$, $\mathbf{n}'=\varsigma^{-1} \mathbf{n}$, $\mathbf{m}'=e^{i\varpi}\mathbf{m}$, $\mathbf{\overline{m}}'=e^{-i\varpi}\mathbf{\overline{m}}$ with
\begin{equation}\label{regularizing}
	\varsigma = \frac{\abs{\Delta}}{r_+^3\sqrt{2\Sigma}}, \qquad e^{i\varpi}=\frac{\sqrt{\Sigma}}{\rho-ip},
\end{equation}
we obtain the frame
\begin{equation*}
	\begin{split}
		\mathbf{l}'&=\frac{1}{r_+^3\sqrt{2\Sigma}}\left((\rho^2+a^2)(\rho^2+b^2)\partial_t + \Delta \partial_\rho + a(\rho^2+b^2)\partial_\varphi + b(\rho^2+a^2)\partial_\psi\right)\\
		\mathbf{n}'&=\frac{r_+^3}{r^2\Delta\sqrt{2\Sigma}}\left((\rho^2+a^2)(\rho^2+b^2)\partial_t - \Delta \partial_\rho + a(\rho^2+b^2)\partial_\varphi + b(\rho^2+a^2)\partial_\psi\right)\\
		\mathbf{m}'&=\frac{1}{\sqrt{2\Sigma}}\left(\partial_\vartheta + i \frac{\sin\vartheta\cos\vartheta}{p}\left((a^2-b^2)\partial_t + \frac{a}{\sin^2\vartheta}\partial_\varphi - \frac{b}{\cos^2\vartheta}\partial_\psi\right)\right)\\
		\mathbf{\overline{m}}'&=\frac{1}{\sqrt{2\Sigma}}\left(\partial_\vartheta - i \frac{\sin\vartheta\cos\vartheta}{p}\left((a^2-b^2)\partial_t + \frac{a}{\sin^2\vartheta}\partial_\varphi - \frac{b}{\cos^2\vartheta}\partial_\psi\right)\right)\\
		\mathbf{k}'&=\mathbf{k}.
	\end{split}
\end{equation*}
Transforming to Eddington--Finkelstein-type coordinates, the above frame takes the form
\begin{equation*}
	\begin{split}
		\mathbf{l}'&=\frac{1}{r_+^3\sqrt{2\Sigma}}\left((\Delta+4Mr^2)\partial_\tau + \Delta\partial_r + 2a(r^2+b^2)\partial_\phi + 2b(r^2+a^2)\partial_\xi\right)\\
		\mathbf{n}'&=\frac{r_+^3}{r^2\sqrt{2\Sigma}}\left(\partial_\tau - \partial_r\right)\\
		\mathbf{m}'&=\frac{1}{\sqrt{2\Sigma}}\left(\partial_\theta + i \frac{\sin\theta\cos\theta}{p}\left((a^2-b^2)\partial_\tau + \frac{a}{\sin^2\theta}\partial_\phi - \frac{b}{\cos^2\theta}\partial_\xi\right)\right)\\
		\mathbf{\overline{m}}'&=\frac{1}{\sqrt{2\Sigma}}\left(\partial_\theta - i \frac{\sin\theta\cos\theta}{p}\left((a^2-b^2)\partial_\tau + \frac{a}{\sin^2\theta}\partial_\phi - \frac{b}{\cos^2\theta}\partial_\xi\right)\right)\\
		\mathbf{k}'&=\frac{1}{rp}(ab\partial_\tau + b\partial_\phi + a\partial_\xi),
	\end{split}
\end{equation*}
where $M=\mu/2$ is the mass of the black hole. We construct a frame $\partial_A$ satisfying $(\partial_A, \partial_B)=\eta_{AB}$ (i.e. an orthonormal frame) by
\begin{equation*}
	\partial_0 = \frac{\mathbf{l}'+\mathbf{n}'}{\sqrt{2}}, \quad \partial_1 = \frac{\mathbf{l}'-\mathbf{n}'}{\sqrt{2}}, \quad \partial_2=\frac{\mathbf{m}'+\mathbf{\overline{m}'}}{\sqrt{2}}, \quad \partial_3=\frac{\mathbf{m}'-\mathbf{\overline{m}'}}{i\sqrt{2}}, \quad \partial_5=\mathbf{k}'.
\end{equation*}
Explicitly, it is written
\begin{equation}\label{framevectors}
	\begin{split}
		\partial_0&=\frac{1}{2r_+^3\sqrt{\Sigma}}\left(\left(\Delta + 4Mr^2+\frac{r_+^6}{r^2}\right)\partial_\tau + \left(\Delta - \frac{r_+^6}{r^2}\right)\partial_r + 2a(r^2+b^2)\partial_\phi + 2b(r^2+a^2)\partial_\xi\right)\\
		\partial_1&=\frac{1}{2r_+^3\sqrt{\Sigma}}\left(\left(\Delta + 4Mr^2-\frac{r_+^6}{r^2}\right)\partial_\tau + \left(\Delta + \frac{r_+^6}{r^2}\right)\partial_r + 2a(r^2+b^2)\partial_\phi + 2b(r^2+a^2)\partial_\xi\right)\\
		\partial_2&=\frac{1}{\sqrt{\Sigma}}\partial_\theta\\
		\partial_3&=\frac{\sin\theta\cos\theta}{p\sqrt{\Sigma}}\left((a^2-b^2)\partial_\tau + \frac{a}{\sin^2\theta}\partial_\phi - \frac{b}{\cos^2\theta}\partial_\xi\right)\\
		\partial_5&=\frac{1}{rp}(ab\partial_\tau + b\partial_\phi + a\partial_\xi).
	\end{split}
\end{equation}
The corresponding dual orthonormal 1-forms are
\begin{equation}\label{frame1forms}
	\begin{split}
		e^0=&\frac{1}{2r^2r_+^3\sqrt{\Sigma}}\big( (r^2\Delta + r_+^6)d\tau + (-\Delta r^2 - 4Mr^4 + 2r^4(a^2\sin^2\theta + b^2\cos^2\theta) + 2a^2b^2r^2+r_+^6)dr\\
		&- a\sin^2\theta(r^2\Delta + r_+^6)d\phi - b\cos^2\theta (r^2\Delta + r_+^6)d\xi\big)\\
		e^1=&\frac{1}{2r^2r_+^3\sqrt{\Sigma}}\big( (-r^2\Delta + r_+^6)d\tau + (\Delta r^2 - 4Mr^4 - 2r^4(a^2\sin^2\theta + b^2\cos^2\theta) - 2a^2b^2r^2+r_+^6)dr\\
		&- a\sin^2\theta(-r^2\Delta + r_+^6)d\phi - b\cos^2\theta (-r^2\Delta + r_+^6)d\xi\big)\\
		e^2=&\Sigma d\theta\\
		e^3=&\frac{\sin\theta\cos\theta}{p\sqrt{\Sigma}}\left((b^2-a^2)d\tau + (b^2-a^2)dr + a(r^2+a^2)d\phi - b(r^2+b^2)d\xi\right)\\
		e^5=&\frac{1}{rp}\left(-ab d\tau - ab dr + (r^2+a^2)b\sin^2\theta d\phi + (r^2+b^2)a\cos^2\theta d\xi\right).
	\end{split}
\end{equation}
Using (\ref{cartanfirst}), the connection coefficients ${\omega^A}_B$ for the frame (\ref{frame1forms}) are computed to be
\begin{equation*}
	\begin{split}
		{\omega^0}_1&=\partial_r \left( \frac{r^2\Delta + r_+^6}{2r^2r_+^3\sqrt{\Sigma}}\right)e^0 - \partial_r \left( \frac{-r^2\Delta + r_+^6}{2r^2r_+^3\sqrt{\Sigma}}\right)e^1 - \frac{(a^2-b^2)r\sin\theta\cos\theta}{p\Sigma^{3/2}}e^3 - \frac{ab}{r^2p}e^5\\
		{\omega^0}_2&= -\frac{(a^2-b^2)\sin\theta\cos\theta}{\Sigma^{3/2}}e^0+\left(\frac{r^2\Delta - r_+^6}{2rr_+^3\Sigma^{3/2}}\right)e^2- \frac{p(r^2\Delta+r_+^6)}{2r^2r_+^3\Sigma^{3/2}}e^3\\
		{\omega^0}_3&=- \frac{(a^2-b^2)r\sin\theta\cos\theta}{p\Sigma^{3/2}}e^1 + \frac{p(r^2\Delta+r_+^6)}{2r^2r_+^3\Sigma^{3/2}}e^2 + \left(\frac{r^2\Delta - r_+^6}{2rr_+^3\Sigma^{3/2}}\right)e^3\\
		{\omega^0}_5 &= -\frac{ab}{r^2p}e^1 + \left(\frac{r^2\Delta - r_+^6}{2r^3r_+^3 \sqrt{\Sigma}}\right)e^5\\
		{\omega^1}_2 &= -\frac{(a^2-b^2)\sin\theta\cos\theta}{\Sigma^{3/2}}e^1 - \left(\frac{r^2\Delta + r_+^6}{2rr_+^3\Sigma^{3/2}}\right)e^2 - \frac{p(-r^2\Delta+r_+^6)}{2r^2r_+^3\Sigma^{3/2}}e^3\\
		{\omega^1}_3 &= - \frac{(a^2-b^2)r\sin\theta\cos\theta}{p\Sigma^{3/2}}e^0 + \frac{p(-r^2\Delta+r_+^6)}{2r^2r_+^3\Sigma^{3/2}}e^2 - \left(\frac{r^2\Delta + r_+^6}{2rr_+^3\Sigma^{3/2}}\right)e^3\\
		{\omega^1}_5 &= -\frac{ab}{r^2p}e^0 -\left(\frac{r^2\Delta + r_+^6}{2r^3r_+^3 \sqrt{\Sigma}}\right)e^5\\
		{\omega^2}_3 &= -\frac{p(r^2\Delta+r_+^6)}{2r^2r_+^3\Sigma^{3/2}}e^0 + \frac{p(-r^2\Delta+r_+^6)}{2r^2r_+^3\Sigma^{3/2}}e^1 - \frac{p}{\sin\theta\cos\theta}\partial_\theta \left(\frac{\sin\theta\cos\theta}{p\sqrt{\Sigma}}\right)e^3 - \frac{ab}{rp^2}e^5\\
		{\omega^2}_5 &= -\frac{ab}{rp^2}e^3 + \frac{(a^2-b^2)\sin\theta\cos\theta}{p^2\sqrt{\Sigma}}e^5\\
		{\omega^3}_5 &= \frac{ab}{rp}e^2.
	\end{split}
\end{equation*}
Recalling that in our chosen metric signature $\omega_{0A}=-{\omega^0}_A$ and $\omega_{jk}={\omega^j}_k$ for spatial indices $j,k\neq 0$, the coefficients of the spinor connection $\Gamma_Ae^A=\frac{1}{2}\sum_{A<B}\gamma^A\gamma^B\omega_{AB}$ are thus

\begin{equation*}
	\begin{split}
		\Gamma_0 =&\, \frac{1}{2} \Biggl( -\partial_r \left( \frac{r^2\Delta + r_+^6}{2r^2r_+^3\sqrt{\Sigma}}\right)\gamma^0\gamma^1 + \frac{(a^2-b^2)\sin\theta\cos\theta}{\Sigma^{3/2}}\gamma^0\gamma^2 - \frac{(a^2-b^2)r\sin\theta\cos\theta}{p\Sigma^{3/2}}\gamma^1\gamma^3 \\
		&- \frac{ab}{r^2p}\gamma^1\gamma^5 - \frac{p(r^2\Delta+r_+^6)}{2r^2r_+^3\Sigma^{3/2}}\gamma^2\gamma^3\Biggr)\\
		\Gamma_1 =&\, \frac{1}{2} \Biggl( -\partial_r \left( \frac{-r^2\Delta + r_+^6}{2r^2r_+^3\sqrt{\Sigma}}\right)\gamma^0\gamma^1 - \frac{(a^2-b^2)\sin\theta\cos\theta}{\Sigma^{3/2}}\gamma^1\gamma^2 + \frac{(a^2-b^2)r\sin\theta\cos\theta}{p\Sigma^{3/2}}\gamma^0\gamma^3 \\
		&+ \frac{ab}{r^2p}\gamma^0\gamma^5 + \frac{p(-r^2\Delta+r_+^6)}{2r^2r_+^3\Sigma^{3/2}}\gamma^2\gamma^3\Biggr)\\
		\Gamma_2=&\,\frac{1}{2}\Biggl(-\left(\frac{r^2\Delta - r_+^6}{2rr_+^3\Sigma^{3/2}}\right)\gamma^0\gamma^2 - \frac{p(r^2\Delta+r_+^6)}{2r^2r_+^3\Sigma^{3/2}}\gamma^0\gamma^3 - \left(\frac{r^2\Delta + r_+^6}{2rr_+^3\Sigma^{3/2}}\right)\gamma^1\gamma^2 \\
		&+ \frac{p(-r^2\Delta+r_+^6)}{2r^2r_+^3\Sigma^{3/2}}\gamma^1\gamma^3 + \frac{ab}{rp^2}\gamma^3\gamma^5\Biggr)\\
		\Gamma_3=&\,\frac{1}{2}\Biggl( \frac{(a^2-b^2)r\sin\theta\cos\theta}{p\Sigma^{3/2}}\gamma^0\gamma^1 + \frac{p(r^2\Delta+r_+^6)}{2r^2r_+^3\Sigma^{3/2}}\gamma^0\gamma^2 - \left(\frac{r^2\Delta - r_+^6}{2rr_+^3\Sigma^{3/2}}\right)\gamma^0\gamma^3\\
		&-\frac{p(-r^2\Delta+r_+^6)}{2r^2r_+^3\Sigma^{3/2}}\gamma^1\gamma^2 -\left(\frac{r^2\Delta + r_+^6}{2rr_+^3\Sigma^{3/2}}\right)\gamma^1\gamma^3 - \frac{p}{\sin\theta\cos\theta}\partial_\theta \left(\frac{\sin\theta\cos\theta}{p\sqrt{\Sigma}}\right)\gamma^2\gamma^3 - \frac{ab}{rp^2}\gamma^2\gamma^5\Biggr)\\
		\Gamma_5 =&\, \frac{1}{2}\Biggl(\frac{ab}{r^2p}\gamma^0\gamma^1 - \left(\frac{r^2\Delta - r_+^6}{2r^3r_+^3 \sqrt{\Sigma}}\right)\gamma^0\gamma^5 - \left(\frac{r^2\Delta + r_+^6}{2r^3r_+^3 \sqrt{\Sigma}}\right)\gamma^1\gamma^5\\
		&- \frac{ab}{rp^2}\gamma^2\gamma^3 + \frac{(a^2-b^2)\sin\theta\cos\theta}{p^2\sqrt{\Sigma}}\gamma^2\gamma^5\Biggr).
	\end{split}
\end{equation*}

\pagebreak

\bibliography{refs}
\bibliographystyle{amsplain}
\end{document}